\documentclass{article}

\usepackage{amssymb, amsmath, amsthm, graphicx, enumerate, color}
\usepackage[vlined,noend,linesnumbered]{algorithm2e}
\usepackage{breqn}
\usepackage{cancel}

\usepackage{adjustbox}
\usepackage{array}
\usepackage{url}
\usepackage{hyperref}

\newtheorem{theorem}{Theorem}
\newtheorem{remark}{Remark}
\newtheorem{lemma}{Lemma}
\newtheorem{proposition}{Proposition}

\DeclareMathOperator*{\argmax}{arg\,max}
\newcommand{\card}[1]{\lvert#1\rvert}
\newcommand{\setMinus}[2]{#1\setminus#2}
\newcommand{\size}[1]{\lvert#1\rvert}
\newcommand{\condi}[0]{\mid}

\newcommand{\staOpt}[1]{\mathnormal{#1}}
\newcommand{\init}[0]{\staOpt{o}}
\newcommand{\stateset}[0]{\staOpt{Q}}

\newcommand{\prematchset}[0]{\staOpt{F}}

\newcommand{\cState}[0]{q}

\newcommand{\cPos}[0]{p}

\newcommand{\nextmm}[0]{\boldsymbol\alpha}
\newcommand{\trans}[0]{\boldsymbol\delta}
\newcommand{\shift}[0]{\boldsymbol\gamma}

\newcommand{\alp}[0]{\mathcal{A}}
\newcommand{\N}[0]{\mathrm{N}}

\newcommand{\piid}[0]{\pi}
\newcommand{\powerset}[1]{\mathcal{P}(#1)}
\newcommand{\strat}[0]{S}

\newcommand{\defi}[1]{\textit{#1}}

\newcommand{\matchine}[0]{\Gamma}

\newcommand{\prob}[0]{p}

\newcommand{\model}[0]{\mathcal{M}}
\newcommand{\okw}[1][w]{$w$-consistent}

\newcommand{\tac}[1]{a_{#1}}
\newcommand{\as}[2]{\mathrm{A\!S}_{#1}(#2)}

\newcommand{\alg}[0]{\mathbf{A}}
\newcommand{\proba}[1]{\mathrm{p}_{#1}}

\newcommand{\Minit}[1][M]{{\pi}_{#1}}
\newcommand{\Mtrans}[1][M]{{\delta}_{#1}}

\newcommand{\ordmatchine}[1][\matchine]{O_{#1}}
\newcommand{\readset}[1]{R_{#1}}
\newcommand{\sink}{\odot}
\newcommand{\expan}[1]{#1^{\star}}

\newcommand{\mem}[2][\matchine]{\mathbf{h}_{#1}(#2)}

\newcommand{\kshifted}[2]{\stackrel{\scriptscriptstyle#1}{\overleftarrow{#2}}}

\newcommand{\typa}[1]{\dot{#1}}

\newcommand{\sta}[1][q]{\typa{#1}}

\newcommand{\set}[1]{\mathcal{#1}}

\newcommand{\frst}[0]{\mathbf{f}}

\newcommand{\stl}[0]{s}
\newcommand{\lattice}[1][w]{L^{[#1]}}
\newcommand{\stateLattice}[1][w]{Q^{[#1]}}
\newcommand{\transLattice}[2][w]{\delta^{[#1]}_{#2}}
\newcommand{\shiftLattice}[2][w]{\gamma^{[#1]}_{#2}}
\newcommand{\shiftSet}[2]{\Delta({#1},{#2})}
\newcommand{\symdif}{\mathop{\ominus}}
\newcommand{\leql}{\mathop{\preccurlyeq}}

\newcommand{\lsl}{\mathop{\prec}}

\newcommand{\precid}[1][w]{\mathrm{prec}_{#1}}
\newcommand{\precocc}[3][w]{\mathrm{prec}_{#1}[#2,#3]}
\newcommand{\nul}{\mathrm{NULL}}
\newcommand{\order}[1]{O(#1)}
\newcommand{\mapnext}[0]{\phi}
\newcommand{\complementary}[1]{\overline{#1}}
\newcommand{\nAlgos}[0]{9}
\newcommand{\expstate}[1]{E(#1)}

\title{Designing optimal- and fast-on-average pattern matching algorithms}
\author{Gilles Didier and Laurent Tichit\\
\small Aix-Marseille Universit\'e, CNRS, Centrale Marseille, I2M UMR7373, Marseille, France\\ 
\small E-mail: \url{{gilles.didier, laurent.tichit}@univ-amu.fr}}

\begin{document}
\maketitle
\begin{abstract}
Given a pattern $w$ and a text $t$, the speed of a pattern matching algorithm over $t$ with regard to $w$, is the ratio of the length of $t$ to the number of text accesses performed to search $w$ into $t$.
We first propose a general method for computing the limit of the expected speed of pattern matching algorithms, with regard to $w$, over iid texts. Next, we show how to determine the greatest speed which can be achieved among a large class of algorithms, altogether with an algorithm running this speed. Since the complexity of this determination makes it impossible to deal with patterns of length greater than 4, we propose a polynomial heuristic. Finally, our approaches are compared with \nAlgos{} pre-existing pattern matching algorithms from both a theoretical and a practical point of view, i.e.  both in terms of limit expected speed on iid texts, and in terms of observed average speed on real data. In all cases, the pre-existing algorithms are outperformed.
\end{abstract}
\section{Introduction}

We focus on algorithms solving the online string matching problem, which consists in reporting all, and only the occurrence positions of a pattern $w$ in a text $t$ (\defi{online} meaning that no pre-processing of the text is allowed).
As one of the oldest problems addressed in computer science, it has been extensively studied. We refer to \cite{Faro2013} for a comprehensive list and an evaluation of all the pattern matching algorithms developed so far. By the authors' count, more than 80 algorithms have already been proposed, among which more than a half were published during the last ten years. This fact sounds quite paradoxical, since the Morris-Pratt algorithm, which is optimal in terms of worst case analysis, dates back to 1970.

A possible explanation is that there is wide gap between the worst case complexity of algorithms and their computation times on real data. For instance, there are pattern matching algorithms with non-linear worst case complexities, which perform much better than Morris-Pratt on English texts. Basically, the average case analysis is way more suited to assess the relevance of a pattern matching algorithm from a practical point of view. The average case analysis of some pattern matching algorithms, notably Boyer-Moore-Horspool and Knuth-Morris-Pratt, has already been carried out from various points of view \cite{Yao1979,Guibas1981,Barth1984,BaezaYates1992,Mahmoud1997,Regnier1998,Smythe2001,Tsai2006}. We provide here a general method for studying the limit average behavior of a pattern algorithm over iid texts. More precisely, following \cite{Marschall2008}, we consider the limit expectation of the ratio of the text length to the number of text accesses performed by an algorithm for searching a pattern $w$ in iid texts. This limit expectation is called the asymptotic speed of the algorithm with regard to $w$ under the iid model. The computation of the asymptotic speed is based on  $w$-matching machines which are automata-like structures able to simulate the behavior of a pattern matching algorithm while searching the pattern $w$. The underlying idea is the same as in \cite{Marschall2008,Marschall2010,Marschall2011,Marschall2012} and can be seen as a generalization of the string matching automaton \cite{Cormen1990}.

In the companion paper, G. Didier provided a theoretical analysis of the asymptotic speed of pattern matching algorithms over iid texts \cite{DidierX}.
In particular, he showed that, for a given pattern $w$, the greatest asymptotic speed among a large class of pattern matching algorithms, is achieved by a $w$-matching machine in which the states are essentially subsets of positions of $w$. Such machines are called \defi{strategies} below.

We provide here a brute force algorithm computing the \defi{Fastest} strategy for a given pattern $w$ and the frequencies of an iid model.
The algorithm is based on an original structure associated to the pattern $w$ and called its position lattice, which gives a full representation of the overlap relations between the subsets of positions of $w$.

Since the brute force algorithm cannot be applied on patterns of length greater than $4$, because of its (very high) time-complexity, we propose a polynomial \defi{$K$-Heuristic}, in which the polynomial order $K$ may be chosen by the user.

The Fastest and $K$-Heuristic approaches are finally compared with \nAlgos{} several pre-existing pattern matching algorithms:
\begin{itemize}
 \item from a theoretical point of view, by computing their limit expected speeds with regard to various patterns and iid models,
 \item from a practical point of view, by computing their average speeds over two sources (an English text and a DNA sequence).
\end{itemize}

In both cases, the Fastest and $K$-Heuristic (with $K$ large enough) approaches outperform the pre-existing algorithms.

The software and the data used to perform the tests are available at \url{https://github.com/gilles-didier/Matchines.git}.

The rest of the paper is organized as follows. Section \ref{secNotation} presents the notations and recalls some concepts and results from \cite{DidierX}. It is followed by two sections which introduce the central objects of this work: the strategies and the position lattice of a pattern. In particular, we provide an algorithm computing the position lattice of a given pattern. Section \ref{secBrute} shows how to use the position lattice of a pattern to obtain the Fastest strategy with regard to this pattern and an iid model. In Section  \ref{secHeuristic}, we provide a polynomial heuristic allowing to compute fast strategies. Section \ref{secEvaluation} presents the results of various comparisons between \nAlgos{} pre-existing pattern matching algorithms, the $K$-Heuristic and, each time it is possible, the Fastest strategy. The results are discussed in the last section.

\section{Notations and definitions}\label{secNotation}
\subsection{Notations and general definition}
For all finite sets $\set{S}$, $\powerset{\set{S}}$ is the power set of $\set{S}$ and $\card{\set{S}}$ is its cardinal.
An \defi{alphabet} is a finite set $\alp$ of elements called \defi{letters} or \defi{symbols}.

A \defi{word}, a \defi{text} or a \defi{pattern} on $\alp$ is a finite sequence of symbols of $\alp$.
%We put $\size{v}$ for the length of a word $v$ and $\size{v}_{a}$ for the number of occurrences of the symbol $a$ in $v$. 
We put $\size{v}$ for the length of a word $v$.
Words are indexed from $0$, i.e. $v = v_{0}v_{1}\ldots v_{\size{v}-1}$. We write $v_{[i,j]}$ for the subword of $v$ starting at its position $i$ and ending at its position $j$, i.e. $v_{[i,j]} = v_{i}v_{i+1}\ldots v_{j}$. The \defi{concatenate} of two words $u$ and $v$ is the word $uv=u_{0}u_{1}\ldots u_{\size{u}-1}v_{0}v_{1}\ldots v_{\size{v}-1}$.

For any length $n\geq 0$, we note $\alp^{n}$ the set of words of length $n$ on $\alp$, and $\alp^{\star}$, the set of finite words on $\alp$, i.e. $\alp^{\star} = \bigcup_{n=0}^{\infty} \alp^{n}$.

Unless otherwise specified, all the texts and patterns considered below are on a fixed alphabet $\alp$.

A \defi{pattern matching algorithm} takes a pattern $w$ and a text $t$ as inputs an reports all, and only the occurrence positions of $w$ in $t$.
For all patterns $w$, we say that two pattern matching algorithms are \defi{$w$-equivalent} if, for all texts $t$, they access exactly the same positions of $t$ on the input $(w,t)$.

\subsection{Matching machines and the generic algorithm \cite{DidierX}}
For all patterns $w$, a \defi{$w$-matching machine} is $6$-uple $(\stateset, \init, \prematchset, \nextmm, \trans, \shift)$ where

\begin{itemize}
\item $\stateset$ is a finite set of states,
\item $\init\in\stateset$ is the initial state,
\item $\prematchset\subset\stateset$ is the subset of pre-match states,
\item $\nextmm:\stateset\rightarrow \mathbb{N}$ is the next-position-to-check function, which is such that for all $q\in\prematchset$, $\nextmm(q)<\size{w}$,
\item $\trans:\stateset\times\alp\rightarrow \stateset$ is the transition state function,
\item $\shift:\stateset\times\alp\rightarrow \mathbb{N}$ is the shift function.
\end{itemize}

By convention, the set of states of a matching machine always contains a \defi{sink state} $\sink$, which is such that, for all symbols $x\in\alp$,  $\trans(\sink, x) = \sink$ and $\shift(\sink, x) = 0$.
% comme on vient de parler de  $\trans(\sink, x)$ et de  $\shift(\sink, x)$, peut-etre parler de $\nextmm(\sink)$ ? par exemple indiquer que $\nextmm(\sink)$ n'est pas défini ?

The \defi{order} $\ordmatchine$ of a matching machine $\matchine=(\stateset, \init, \prematchset, \nextmm, \trans, \shift)$ is defined as $\ordmatchine = \max_{q\in\stateset}\{\nextmm(q)\}$.

The $w$-matching machines carry the same information as the \defi{Deterministic Arithmetic Automatons} defined in \cite{Marschall2010,Marschall2011}.

The generic algorithm takes a $w$-matching machine and a text $t$ as inputs and outputs  positions of $t$ (Algorithm \ref{algogen}).

\begin{algorithm}[htp]
\SetKw{and}{and}
\SetKw{print}{print}
\SetKwInOut{Input}{input}
\SetKwInOut{Output}{output}
	\SetAlgoLined\DontPrintSemicolon
	\Input{a $w$-matching machine $(\stateset, \init, \prematchset, \nextmm, \trans, \shift)$ and a text $t$}
	\Output{all the occurrence positions of $w$ in $t$}
	\vskip 0.2cm	
		$(\cState, \cPos) \leftarrow (\init, 0)$\;
		\While{$\cPos\leq\size{t}-\size{w}$} {
			\If{$\cState\in\prematchset$ \and $t_{\cPos+\nextmm(\cState)} = w_{\nextmm(\cState)}$\nllabel{liTestGA}}{
				\print ``~occurrence at position $\cPos$~''\;
			}
			$(\cState, \cPos) \leftarrow (\trans(\cState, t_{\cPos+\nextmm(\cState)}), \cPos+\shift(\cState, t_{\cPos+\nextmm(\cState)}))$\;
		}
		\vskip 0.3cm
\caption{The generic algorithm}\label{algogen}
\end{algorithm} 

Each component of a $w$-matching machine makes sense in regard to the way it is used by the generic algorithm. The pre-match states in $\prematchset$ are those which lead to report an occurrence of the pattern at the current position, if the next-position-to-check of the pattern matches the corresponding position in the text (Line \ref{liTestGA} of Algorithm \ref{algogen}). The condition  $\nextmm(q)<\size{w}$ for all  $q\in\prematchset$ in the definition of $w$-matching machines, is technical and used in \cite{DidierX}.

A $w$-matching machine $\matchine$ is \defi{valid} if, for all texts $t$, the execution of the generic algorithm on the input $(\matchine, t)$ outputs all, and only the occurrence positions of $w$ in $t$. Since one has to check all the positions of the pattern $w$ before concluding that it occurs somewhere in a text, the order of a valid $w$-matching machine is at least $\size{w}-1$.

We claim that for all the pattern matching algorithms developed so far and all patterns $w$, there exists a $w$-matching machine $\matchine$ which is such that, for all texts $t$, the generic algorithm and the pattern matching algorithm access exactly the same positions of $t$ on the inputs $(\matchine, t)$ and $(w, t)$ respectively \cite{DidierX}.
For instance, Figure \ref{figMatchineBas} displays a $abb$-matching machine which accesses the same positions as the naive algorithm while searching $abb$.

\begin{figure*}[!tpb]%figure1
\centering{\includegraphics[width=0.75\textwidth, trim=0cm -0.0cm 0cm 0cm]{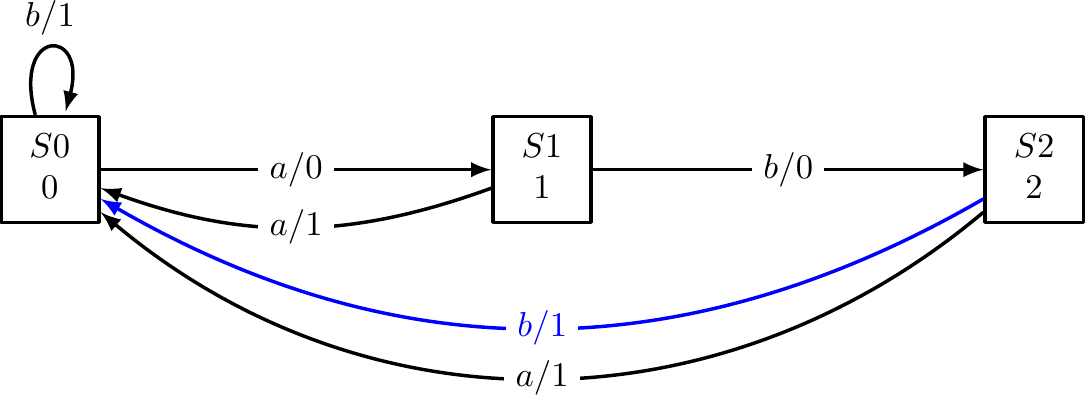}}
\caption{$abb$-matching machine of the naive algorithm. The next-position-to check are displayed below all states $S0$, $S1$ and $S2$. Edges from states $Si$ are labelled with $``x/\shift(Si, x)"$ for all symbols $x$. The  transition associated to a match is blue-colored.}
\label{figMatchineBas}
\end{figure*}

\subsection{Full-memory expansion -- standard matching machines \cite{DidierX}}\label{secFull}

%Cette partie arrive ici comme un cheveu sur la soupe, manque une phrase d'intro
%Il faut expliquer pourquoi on a besoin de la full-memory expansion: exemple de phrase d'intro:
%In order to design efficient mathcing patching, it is necessary to keep track of the information of the text position already visited thus far. The full-memory expansion of a matching machine stores in each state, the information on the text.
We present here a transformation on matching machines which split their states according to the text positions read from the current position during an execution of the generic algorithm. The main point of this transformation is that the average complexity of matching machines such obtained may then be computed through algebraic methods (Sections \ref{secModels} and \ref{secSpeed}).

For all $n\in\N$, $\readset{n}$ is the set of subsets $H$ of  $\{0,\ldots, n\}\times\alp$ verifying that, for all $i\in\{0, \ldots, n\}$, there exists at most one pair in $H$ with $i$ as first entry.%, i.e. $\readset{n}$ is the set of partial functions from $\{0, \ldots, n\}$ to $\alp$.

For all $H\in\readset{n}$, we put $\frst(H)$ for the set comprising all the first entries of the pairs in $H$, namely  
\begin{dmath*}\frst(H)=\{i\hiderel{\condi} \exists x\hiderel{\in}\alp\mbox{ with }(i,x)\hiderel{\in} H\}.\end{dmath*} 

For all $k\in\N$ and $H\in\readset{n}$, the \defi{$k$-shifted} of $H$ is 
\begin{dmath*}
\kshifted{k}{H} \hiderel{=} \{(u-k,y) \hiderel{\condi}  (u, y)\hiderel{\in} H \mbox{ with } u\hiderel{\geq} k\},\end{dmath*}
i.e. the subset of $\readset{n}$ obtained by subtracting $k$ from the first entries of the pairs in $H$ and by keeping only the pairs with non-negative first entries.

The \defi{full memory expansion} of a $w$-matching machine $\matchine = (\stateset, \init, \prematchset, \nextmm, \trans, \shift)$ is the $w$-matching machine $\expan{\matchine}$ obtained by removing the unreachable states of $\matchine' = (\stateset', \init', \prematchset', \nextmm', \trans', \shift')$, defined as:
\begin{itemize}
\item $\stateset' = \stateset\times\readset{\ordmatchine}$
\item $\init' = (\init, \emptyset)$
\item $\nextmm'((q,H)) = \nextmm(q)$
\item $\shift'((q,H), x) = \shift(q, x)$
\item $\prematchset' = \prematchset\times\readset{\ordmatchine}$
\item \begin{dmath*}\trans'((q,H), x) = \left\{\begin{array}{ll}
(\trans(q,x),  \kshifted{\shift(q, x)}{H\cup\{(\nextmm(q),x)\}}) & \mbox{if $\forall a\in\alp, (\nextmm(q),a)\not\in H$}\\%\mbox{if there is no pair $(\nextmm(q),a)\in H$} \\
\sink & \mbox{if $\exists a\neq x$ s.t. $(\nextmm(q),a)\in H$}\\
(\trans(q,x),  \kshifted{\shift(q, x)}{H}) & \mbox{if $(\nextmm(q),x)\in H$}
\end{array}\right.\end{dmath*}
\end{itemize}

\begin{figure*}[!tpb]%figure1
\centering{\includegraphics[width=0.75\textwidth, trim=0cm -0.0cm 0cm 0cm]{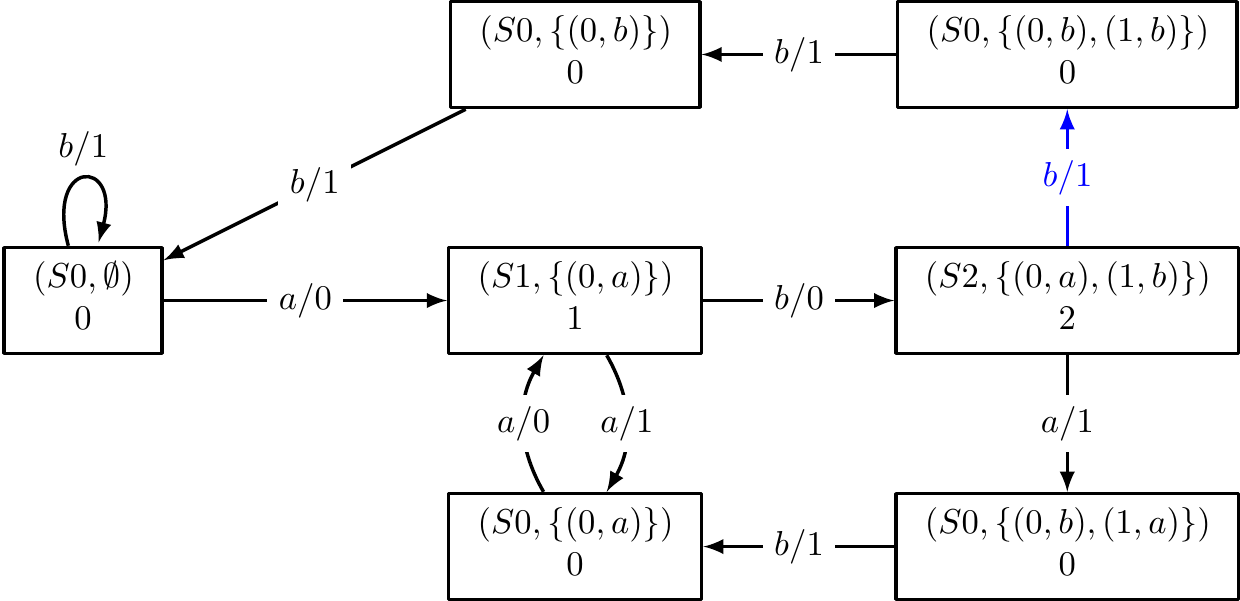}}
\caption{Full memory expansion of the $abb$-matching machine of Figure \ref{figMatchineBas}. %Transitions going to the sink state are not displayed.
}
\label{figMatchineExt}
\end{figure*}

By construction, at all iterations of the generic algorithm on the input $(\expan{\matchine}, t)$, if the current state and position are  $(q, H)$ and  $p$, respectively, then the positions of $\{j+p\condi  j\in\frst(H)\}$ are exactly the positions of $t$ greater than  $p$ accessed so far (the second entries of the corresponding elements of $H$ give the symbols read).

For all texts $t$,  the generic algorithm access the same positions of $t$ on the inputs $(\matchine, t)$ and $(\expan{\matchine}, t)$ \cite{DidierX}.

Let us remark that the full memory expansion of  the full memory expansion of a matching machine is equal to its full memory expansion (up to a state isomorphism).
%This transformation is idempotent
A $w$-matching machine $\matchine$ is \defi{standard} if each state $q$ of $\matchine$ appears in a unique pair/state of its full memory expansion, or, equivalently, if it is equal to its full memory expansion. For instance the $abb$-matching machine of Figure \ref{figMatchineBas} is not standard. Since the matching machine of Figure \ref{figMatchineExt} is a full memory expansion, it is standard. For all states $q$ of a standard matching machine $\matchine$, we put $\mem{q}$ for the second entry of the unique pair/state of $\expan{\matchine}$ in which $q$ appears.

We implemented a basic algorithm computing the full-memory expansion $\expan{\matchine}=(\expan{\stateset}, \expan{\init}, \expan{\prematchset}, \expan{\nextmm}, \expan{\trans}, \expan{\shift})$ of a $w$-matching machine ${\matchine}=({\stateset}, {\init}, {\prematchset}, {\nextmm}, {\trans}, {\shift})$ in  $O(\size{w}.\card{\expan{\stateset}})$ time. We have $\card{\expan{\stateset}}\leq(\alp+1)^{\size{w}}\card{\stateset}$ but the size of $\expan{\stateset}$ may vary a lot with regard to the matching machine/algorithm considered.

A $w$-matching machine $\matchine$ is \defi{compact} if it contains no state $q$ which always leads to the same state. Formally, $\matchine= (\stateset, \init, \prematchset, \nextmm, \trans, \shift)$ is compact if there is no $q\in\stateset$ such that one of the following assertions holds:
\begin{enumerate}
	\item there exists a symbol $x$ with $\trans(\sta,x)\neq\sink$ and $\trans(\sta,y)=\sink$ for all symbols $y\neq x$;
	\item for all symbols $x$ and $y$, we have both $\trans(\sta,x) = \trans(\sta,y)$ and $\shift(\sta,x) = \shift(\sta,y)$.
\end{enumerate}
Basically, a non-compact machine performs useless text accesses. In \cite{DidierX}, it is shown that any $w$-matching machine can be turned into a compact (and faster) machine.

\subsection{iid and Markov models}\label{secModels}

An \defi{independent identically distributed (iid)} model (aka \defi{Bernoulli} model) is fully specified by a probability distribution $\piid$ on the alphabet (i.e. $\piid(x)$ is the probability of the symbol $x$ in the model). Such a model will be simply referred to as ``$\piid$'' below. Under $\piid$, the probability of a text $t$ is 
\begin{dmath*}
\proba{\piid}(t) = \prod_{i=0}^{\size{t}-1} \piid(t_{i}).
\end{dmath*}

A \defi{Markov} model $M$ over a given set of states $Q$ is a $2$-uple $(\Minit, \Mtrans)$, where $\Minit$ is a probability distribution on $Q$ (the initial distribution) and $\Mtrans$ associates a pair of states $(q, q')$ with the probability for $q$ to be followed by $q'$ (the transition probability). Under a Markov model $M = (\Minit, \Mtrans)$, the probability of a sequence $s$ of states  is 
\begin{dmath*}
\proba{M}(s) = \Minit(s_{0}) \prod_{i=0}^{\size{s}-1} \Mtrans(s_{i}, s_{i+1}).
\end{dmath*}

\begin{theorem}[\cite{DidierX}]\label{theoiid}
	Let $\matchine = (\stateset, \init, \prematchset, \nextmm, \trans, \shift)$ be a $w$-matching machine. If a text $t$ follows an iid model and $\matchine$ is standard then the sequence of states parsed by the generic algorithm on the input $(\matchine, t)$ follows a Markov model $(\Minit, \Mtrans)$.
\end{theorem}
\begin{proof}
Whatever the text model and the machine, the sequence of states always starts with $\init$ with probability $1$. We have $\Minit(o)=1$ and $\Minit(q)=0$ for all $q\neq o$.

The probability $\Mtrans(q,q')$ that the state $q'$ follows the state $q$ during an execution of the generic algorithm, is equal to:
\begin{itemize}
\item $1$,  if there exists a symbol $x$ such that $\trans(q,x) = q'$ and  $(\nextmm(q), x)\in\mem{q}$, i.e. if the relative position $\nextmm(q)$ was already checked with $x$ occurring at it,
\item $\displaystyle\sum_{\mbox{$x$  s.t. $\trans(q,x) = q'$}} \piid(x)$, otherwise,
\end{itemize}
 independently of the previous states.
\end{proof}

\subsection{Asymptotic speed}\label{secSpeed}

%In order to make the notations less cluttered, $w$ does not appear neither on $\as{\model}{\alg}$ nor on $\tac{\alg}(t)$, but these two quantities actually depend on $w$.
%At this point, nothing ensures that the limit above exists.
Let $\model$ be a  text model and $\alg$ be an algorithm. The \defi{$w$-asymptotic speed} of $\alg$ under $\model$ is the limit expectation, under $\model$, of the ratio of the text length to the number of text accesses performed by $\alg$ \cite{DidierX}. Namely, by putting $\tac{\alg}(t)$ for the number of text accesses performed by $\alg$ to parse $t$ and $\prob_{\model}(t)$ for the probability of $t$ with regard to $\model$, the asymptotic speed of $\alg$ under $\model$ is
\begin{dmath*}
\as{\model}{\alg} = \lim_{n\rightarrow\infty}\sum_{t\in\alp^{n}} \frac{\size{t}}{\tac{\alg}(t)}\prob_{\model}(t).
\end{dmath*}

The asymptotic speed $\as{\model}{\matchine}$ of a $w$-matching machines $\matchine$ is that the generic algorithm with $\matchine$ as first input.
From Theorem 5 of \cite{DidierX}, the asymptotic speed of a standard $w$-matching machine $\matchine=({\stateset}, {\init}, {\prematchset}, {\nextmm}, {\trans}, {\shift})$ under an iid model $\piid$ exists and is given by
\begin{dmath}\label{eqAsympFreq}
\as{\piid}{\matchine} = \sum_{q\in{\stateset}} \beta_{q}\expstate{q},
\end{dmath}
where $(\beta_{q})_{q\in{\stateset}}$ are the limit frequencies of the states of the Markov model associated to ${\matchine}$ and $\piid$ in Theorem \ref{theoiid}, and 
\begin{dmath*}
\expstate{q} = \left\{\begin{array}{ll}
{\shift}(q,x) & \mbox{if  $(\nextmm(q), x)\in\mem{q}$),}\\
\sum_{x} {\shift}(q,x)\piid_{x} &\mbox{otherwise.}
\end{array}\right.
\end{dmath*}

Computing the asymptotic speed of a pattern matching algorithm, with regard to a pattern $w$ and an iid model $\piid$ is performed by following the stages below.
\begin{enumerate}
\item We get a $w$-matching machine $\matchine$ which simulates the behavior of the algorithm while looking for $w$ (Figure \ref{figMatchineBas}). The transformation of the \nAlgos{} algorithms presented in Section \ref{secEvaluation} (and a few others, see our GitHub repository) into $w$-matching machines, given $w$, has been implemented.
\item We obtain the full-memory expansion $\expan{\matchine}$ of $\matchine$ (Figure \ref{figMatchineExt}, Section \ref{secFull}).
\item We compute the limit frequencies of the Markov model associated to $\expan{\matchine}$ and $\piid$ in Theorem  \ref{theoiid}. This mainly needs to solve a system of linear equations of dimension $\card{\expan{\stateset}}$.
\item We finally obtain the asymptotic speed of the algorithm from these limit frequencies, $\piid$ and $\expan{\matchine}$ by using Equation \ref{eqAsympFreq}.
\end{enumerate} 

The most time-consuming stage is the computation of the limit frequencies, which has $O(\card{\expan{\stateset}}^{3})$ time complexity, where $\card{\expan{\stateset}}$, the number of states of the full memory expansion, is smaller than $(\alp+1)^{\size{w}}\card{\stateset}$.

\section{Strategies}\label{secStrategy}

For all sets $\set{I}\subset\mathbb{N}$ and $k\in\mathbb{N}$, we define the $k$-left-shifted of $\set{I}$ as  
\begin{dmath*}
\shiftSet{\set{I}}{k}\hiderel{=}\{i-k\condi  \exists i\hiderel{\in}\set{I} \mbox{ and } i \hiderel{\geq} k\}.
\end{dmath*}

A \defi{$w$-strategy} $\strat=(\stateset, \init, \prematchset, \nextmm, \trans, \shift)$ is a $w$-matching machine  such that
\begin{itemize}
\item $\stateset\subseteq\setMinus{\powerset{\{0, \ldots, \size{w}-1\}}}{\{\{0, \ldots, \size{w}-1\}\}}$ and $\emptyset\in\stateset$,{}
\item $\init = \emptyset$,{}
\item $\prematchset = \{\stl\in\stateset\condi \card{\stl} = \size{w}-1\}$,
\item $\nextmm:\stateset\rightarrow \{0, 1, \ldots, \size{w}-1\}$ is such that for all $\stl\in\stateset$, $\nextmm(\stl)\not\in \stl$ and $\card{\stl}<\size{w}-1\Rightarrow\stl\cup\{\nextmm(\stl)\}\in\stateset$,
\item $\shift:\stateset\times\alp\rightarrow \{0, 1, \ldots, \size{w}\}$ is such that for all states $\stl$ and all symbols $x$,
\begin{dmath*}\shift(\stl,x)=\left\{\begin{array}{ll}\min\{k\geq 1\condi  w_{\nextmm(\stl)-k} = x \mbox{ if } \nextmm(\stl)\geq k \mbox{ and } w_{j} = w_{j+k} \mbox{ for all } j\in\shiftSet{\stl}{k}\} &\mbox{if } \stl\in\prematchset,\\
\min\{k\geq 0\condi  w_{\nextmm(\stl)-k} = x \mbox{ if } \nextmm(\stl)\geq k \mbox{ and } w_{j} = w_{j+k} \mbox{ for all } j\in\shiftSet{\stl}{k}\} &\mbox{otherwise,}
\end{array}\right.\end{dmath*}
%$$\shift(\stl, x) = \min\{k\geq 0\condi  w_{i-k} = x \mbox{ if } i\geq k \mbox{ and } w_{j} = w_{j+k} \mbox{ for all } j\in\shiftSet{\stl}{k}\},$$
\item $\trans:\stateset\times\alp\rightarrow \stateset$ is such that for all $\stl\in\stateset$ and all symbols $x$,
\begin{dmath*}
	\trans(\stl, x) = \shiftSet{\stl\cup\{\nextmm(\stl)\}}{\shift(\stl, x)}.
\end{dmath*}
\end{itemize}

Figure {\ref{figStrategies}} shows two $abb$-strategies which differ notably in the next-position-to-check of state $\{0\}$.

\begin{figure*}[!tpb]%figure1
\centering{\includegraphics{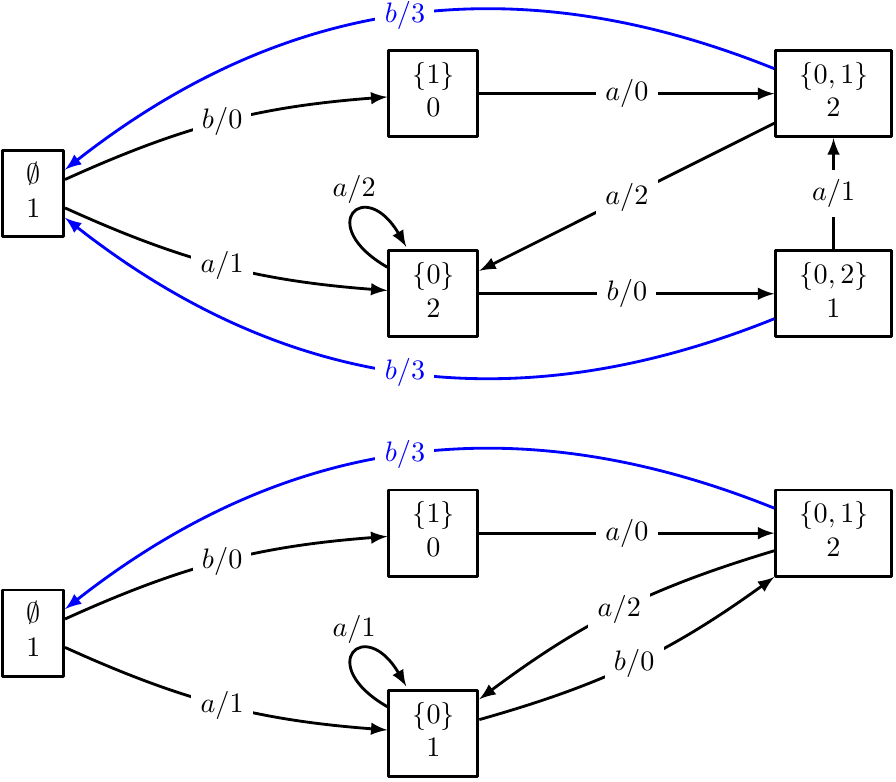}}
\caption{Two $abb$-strategies with the same conventions as in Figure {\ref{figMatchineBas}}.}
\label{figStrategies}
\end{figure*}

\begin{proposition}
A $w$-strategy is a standard, compact, valid and non-redundant $w$-matching machine.
\end{proposition}
\begin{proof}
By construction, a $w$-strategy is standard, compact and non-redundant. The validity of a $w$-strategy follows from Theorem 1 of \cite{DidierX}.
\end{proof}

\begin{proposition}\label{propOptim}
There is a $w$-strategy which achieves the greatest asymptotic speed among all the $w$-matching machines of order $\size{w}-1$.
\end{proposition}
\begin{proof}
The Corollary 2 of  \cite{DidierX} implies that there exists a $w$-matching machine  which achieves the greatest asymptotic speed among those of order $\size{w}-1$ and which is
\begin{enumerate}
%	\item of order $\size{w}-1$, 
	\item standard, 
	\item compact,
	\item valid,
	\item in which all the states are relevant (i.e. such that they may lead to a match without any positive shift \cite{DidierX}),
	\item such that there is no pair of states $(q,q')$ with $q\neq q'$ and $\mem[\matchine_{\piid}]{q} = \mem[\matchine_{\piid}]{q'}$.
\end{enumerate}
Let us verify that a $w$-matching machine $\matchine=(\stateset, \init, \prematchset, \nextmm, \trans, \shift)$ of order $\size{w}-1$ satisfying the properties above is (isomorphic to) a $w$-strategy. Since it verifies in particular the properties 4 and 5, its set of states $\stateset$ is in bijection with a subset of $\powerset{\{0, \ldots, \size{w}-1\}}$. Let us identify all states $q$ of $Q$ with $\frst(\mem{q})$, its corresponding element of $\powerset{\{0, \ldots, \size{w}-1\}}$. Since $\matchine$ is standard, compact and of order $\size{w}-1$, we do not have $\{0, \ldots, \size{w}-1\}\in\stateset$. 
Moreover, since $\matchine$ is standard, we have $\trans(\stl, x) = \shiftSet{\stl\cup\{i\}}{\shift(\stl, x)}$ for all $\stl\in\stateset$.
Last, by construction, if 
\begin{dmath*}\shift(\stl,x)>\left\{\begin{array}{ll}\min\{k\geq 1\condi  w_{\nextmm(\stl)-k} = x \mbox{ if } \nextmm(\stl)\geq k \mbox{ and } w_{j} = w_{j+k} \mbox{ for all } j\in\shiftSet{\stl}{k}\} &\mbox{if } \stl\in\prematchset,\\
\min\{k\geq 0\condi  w_{\nextmm(\stl)-k} = x \mbox{ if } \nextmm(\stl)\geq k \mbox{ and } w_{j} = w_{j+k} \mbox{ for all } j\in\shiftSet{\stl}{k}\} &\mbox{otherwise,}
\end{array}\right.\end{dmath*}
then $\matchine$ is not valid, and if 
\begin{dmath*}\shift(\stl,x)<\left\{\begin{array}{ll}\min\{k\geq 1\condi  w_{\nextmm(\stl)-k} = x \mbox{ if } \nextmm(\stl)\geq k \mbox{ and } w_{j} = w_{j+k} \mbox{ for all } j\in\shiftSet{\stl}{k}\} &\mbox{if } q\in\prematchset,\\
\min\{k\geq 0\condi  w_{\nextmm(\stl)-k} = x \mbox{ if } \nextmm(\stl)\geq k \mbox{ and } w_{j} = w_{j+k} \mbox{ for all } j\in\shiftSet{\stl}{k}\} &\mbox{otherwise,}
\end{array}\right.\end{dmath*}
then $\trans(\stl,x)$ is not relevant.
\end{proof}

\section{Position lattices}\label{secLattice}

The position lattice of a pattern $w$ is the 3-uple $\lattice=(\stateLattice, (\transLattice{s})_{s\in\stateLattice}, (\shiftLattice{s})_{s\in\stateLattice})$ where, by putting $\complementary{\stl}$ for  $\setMinus{\{0, \ldots, \size{w}-1\}}{\stl}$, 
\begin{itemize}
	\item $\stateLattice = \setMinus{\powerset{\{0, \ldots, \size{w}-1\}}}{\{\{0, \ldots, \size{w}-1\}\}}$, i.e. the set made of all the subsets of positions of $w$ but $\{0, \ldots, \size{w}-1\}$,
	\item for all $\stl \in \stateLattice$, $\shiftLattice{\stl}$ is a map from $\complementary{\stl}\times \alp$ to $\{0, \ldots, \size{w}\}$,
	\item for all $\stl \in \stateLattice$, $\transLattice{\stl}$ is a map from $\complementary{\stl}\times \alp$ to $\stateLattice$,
\end{itemize}
where, for all $\stl \in \stateLattice$, all $i\in\complementary{\stl}$ and all $x\in\alp$, we have
\begin{dmath*}\shiftLattice{\stl}(i, x)=\left\{\begin{array}{ll}\min\{k\geq 1\condi  w_{\nextmm(\stl)-k} = x \mbox{ if } \nextmm(\stl)\geq k \mbox{ and } w_{j} = w_{j+k} \mbox{ for all } j\in\shiftSet{\stl}{k}\} &\mbox{if } \card{\stl} = \size{w}-1,\\
\min\{k\geq 0\condi  w_{\nextmm(\stl)-k} = x \mbox{ if } \nextmm(\stl)\geq k \mbox{ and } w_{j} = w_{j+k} \mbox{ for all } j\in\shiftSet{\stl}{k}\} &\mbox{otherwise,}
\end{array}\right.\end{dmath*}
and
\begin{dmath*}
	\transLattice{\stl}(i, x) = \shiftSet{\stl\cup\{i\}}{\shiftLattice{\stl}(i, x)}.
\end{dmath*}
In particular, if $x = w_{i}$ and $\card{\stl} < \size{w}-1$ then we have $\shiftLattice{\stl}(i, x) = 0$ and $\transLattice{\stl}(i, x) = \stl\cup\{i\}$.

Let us remark that, since $\max(\stl)\leq\size{w}-1$ for all $\stl \in \stateLattice$, we have, for all $i\in\complementary{\stl}$ and all $x\in\alp$, 
$\shiftSet{\stl\cup\{i\}}{\size{w}}=\emptyset$, thus  $\shiftLattice{\stl}(i, x)\leq\size{w}$ which is consistent with the definition of $\shiftLattice{\stl}$.

The edges of $\lattice$ are the pairs $(\stl, \transLattice{\stl}(i, x))$ for all $\stl \in \stateLattice$, all $i\in\complementary{\stl}$ and all $x\in\alp$ (see Figure \ref{figLattice}).

\begin{figure*}[!tpb]
\centering{\includegraphics[width=\textwidth, trim=0cm -0.0cm 0cm 0cm]{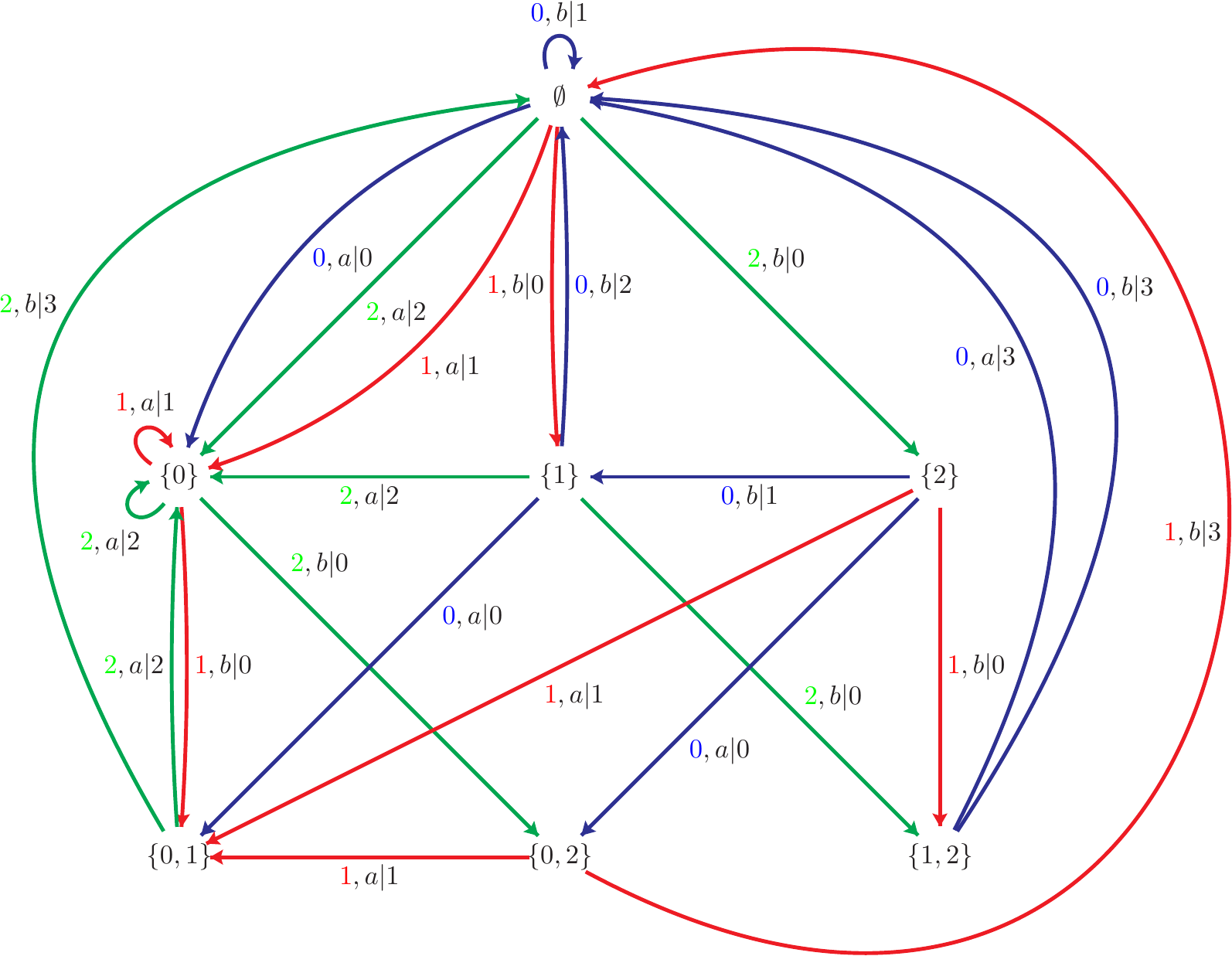}}
\caption{Position lattice of the pattern $abb$. Vertices represent the states of $\lattice[abb]$. For all states $\stl$, there is an outgoing edge for all pairs $(i, x)$ with $i\in\setMinus{\{0, \ldots, \size{abb}-1\}}{\stl}$ and $x\in\alp$. This outgoing edge is labeled with ``$i, x|\shiftLattice[abb]{\stl}(i,x)$'', is colored according to $i$, and goes to $\transLattice[abb]{\stl}(i,x)$.}
\label{figLattice}
\end{figure*}

\begin{remark}\label{remSizeLattice}
The position lattice of $w$ contains $2^{\size{w}}-1$ states and $\card{\alp}.\size{w}.2^{\size{w}-1}$ edges.
\end{remark}

\begin{remark}\label{remComp}
Let $\stl$ be a state of $\stateLattice$, $i$ and $j$ be two positions in $\complementary{\stl}$ such that $i\neq j$ and $x$ and $y$ be two symbols of $\alp$.
We have
\begin{dgroup*}
\begin{dmath*}
\shiftLattice{\transLattice{\stl}(i, x)}(j-\shiftLattice{\stl}(i, x), y)+\shiftLattice{\stl}(i, x) = \shiftLattice{\transLattice{\stl}(j, y)}(i-\shiftLattice{\stl}(j, y), x)+\shiftLattice{\stl}(j, y) \condition{and}
\end{dmath*}
\begin{dmath*}
\transLattice{\transLattice{\stl}(i, x)}(j-\shiftLattice{\stl}(i, x), y) = \transLattice{\transLattice{\stl}(j, y)}(i-\shiftLattice{\stl}(j, y), x).
\end{dmath*}
\end{dgroup*}
By considering the particular case where $x=w_{i}$, we get 
\begin{dgroup*}
\begin{dmath*}
\shiftLattice{\stl\cup\{ i \}}(j, y) = \shiftLattice{\transLattice{\stl}(j, y)}(i-\shiftLattice{\stl}(j, y), w_{i})+\shiftLattice{\stl}(j, y)\condition{and}
\end{dmath*}
\begin{dmath*}
\transLattice{\stl\cup\{ i \}}(j, y) = \transLattice{\transLattice{\stl}(j, y)}(i-\shiftLattice{\stl}(j, y), w_{i}).
\end{dmath*}
\end{dgroup*}
\end{remark}

Let $\precid$ be the table indexed on $\{0, \ldots, \size{w}-1\}\times\alp$ and in which, for all positions $i$ of $w$ and all symbols $x$ of $\alp$, the entry $\precocc{i}{x}$ is defined as
\begin{dmath*}
\precocc{i}{x} = \left\{\begin{array}{ll}\max\{j\leq i \condi  w_{j} = x\} & \mbox{if } \{j\leq i \condi  w_{j} = x\}\neq \emptyset,\\
\nul & \mbox{otherwise.}\end{array}\right.
\end{dmath*}
For instance, the table $\precid[abb]$ is 
$$
\begin{array}{r|cc}
 & a & b \\ \hline
0 & 0 & \nul\\
1 & 0 & 1 \\
2 & 0 & 2 
\end{array}
$$

\newcommand{\KwBordMax}{\mathrm{B}}

\begin{lemma}\label{lemmaLattice}
	Let $\stl$ be a state of $\stateLattice$, $i$ a position in $\complementary{\stl}$ and $x$ a symbol of $\alp$.
	\begin{enumerate}
	\item If $x=w_{i}$ then 
    \begin{itemize}
    	\item if $\card{\stl}=\size{w}-1$ then $\transLattice{\stl}(i,x) = \{0,\ldots,\KwBordMax-1\}$ 	and $\shiftLattice{\stl}(i,x) = \size{w}-\KwBordMax$, where $\KwBordMax$ is the length of the longest proper suffix of $w$ which is a prefix of $w$;
    	\item otherwise $\transLattice{\stl}(i,x) = \stl\cup\{i\}$ 	and $\shiftLattice{\stl}(i,x) = 0$.
        \end{itemize}
	\item If $x\neq w_{i}$,
	\begin{enumerate}
	\item if $s = \emptyset$ then
	\begin{dgroup*}
		\begin{dmath*}
			\transLattice{\emptyset}(i,x) = \left\{\begin{array}{ll}
			\{\precocc{i}{x}\} & \mbox{ if } \precocc{i}{x} \neq \nul,\\
			\emptyset & \mbox{ otherwise,}\\
			\end{array}\right.
		\end{dmath*}
		\begin{dmath*}
			\shiftLattice{\emptyset}(i,x) = \left\{\begin{array}{ll}
			i-\precocc{i}{x} & \mbox{ if } \precocc{i}{x} \neq \nul,\\
			i+1 & \mbox{ otherwise,}\\
			\end{array}\right.
		\end{dmath*}
	\end{dgroup*}
	\item if $s\neq\emptyset$ then for all $\ell\in\stl$, we have  
	\begin{dgroup*}
		\begin{dmath*}
			\transLattice{\stl}(i,x) = 
			\transLattice{\transLattice{\setMinus{\stl}{{\{\ell\}}}}(i,x)}(\ell-\shiftLattice{\setMinus{\stl}{{\{\ell\}}}}(i,x), w_{\ell}),
		\end{dmath*}
		\begin{dmath*}
			\shiftLattice{\stl}(i,x) = 
			\shiftLattice{\transLattice{\setMinus{\stl}{{\{\ell\}}}}(i,x)}(\ell-\shiftLattice{\setMinus{\stl}{{\{\ell\}}}}(i,x), w_{\ell}) + \shiftLattice{\setMinus{\stl}{{\{\ell\}}}}(i,x).
		\end{dmath*}
	\end{dgroup*}
	\end{enumerate}
	\end{enumerate}
\end{lemma}
\begin{proof}
The only case which does not immediately follow from the definition of $\lattice$, is when $x\neq w_{i}$ and $s\neq\emptyset$ which is given by Remark \ref{remComp}.
\end{proof}

The relation  $\leql$ on $\stateLattice$ is defined as follows. For all sets $\stl$ and $\stl'$ in $\stateLattice$, we have $\stl\leql\stl'$ if one of the following properties holds:
\begin{itemize}
\item $\card{\stl}<\card{\stl'}$,
\item $\card{\stl}=\card{\stl'}$, $\stl\neq\stl'$ and $\min(\stl\symdif\stl')\in\stl$, where $\stl\symdif\stl'$ is the symmetric difference of $\stl$ and $\stl'$,
\item $\stl=\stl'$.
\end{itemize}
The relation $\leql$ defines a total order on $\stateLattice$. We write ``~$\stl\lsl\stl'$~'' for ``~$\stl\leql\stl'$ and $\stl\neq\stl'$~''.

\begin{lemma}\label{lemmaOrder}
Let $\stl$ be a state of $\stateLattice$ with $\card{\stl}>1$, $i$ a position in $\complementary{\stl}$ and $x$ a symbol of $\alp$.
If $x\neq w_{i}$ then $\transLattice{\setMinus{\stl}{{\{\max\stl\}}}}(i,x)\lsl\stl$.
\end{lemma}

\begin{proof}
	Under the assumption that $\card{\stl}>1$, we have $\min \stl = \min (\setMinus{\stl}{{\{\max\stl\}}})$. 
	By construction, the fact that $x\neq w_{i}$ implies that $\shiftLattice{\setMinus{\stl}{{\{\max\stl\}}}}(i,x)>0$.
	
	If we have $$\min((\setMinus{\stl}{{\{\max\stl\}}})\cup\{i\})<\shiftLattice{\setMinus{\stl}{{\{\max\stl\}}}}(i,x)$$ then 
	$\card{\transLattice{\setMinus{\stl}{{\{\max\stl\}}}}(i,x)}<\card{\stl}$, thus $\transLattice{\setMinus{\stl}{{\{\max\stl\}}}}(i,x)\lsl\stl$.{}
	
	Otherwise, we have  $\card{\transLattice{\setMinus{\stl}{{\{\max\stl\}}}}(i,x)}=\card{\stl}$ but since necessarily $$\min\transLattice{\setMinus{\stl}{{\{\max\stl\}}}}(i,x)\leq \min \card{\stl}-\shiftLattice{\setMinus{\stl}{{\{\max\stl\}}}}(i,x)<\min \card{\stl},$$ we get again $\transLattice{\setMinus{\stl}{{\{\max\stl\}}}}(i,x)\lsl\stl$.
\end{proof}

\newcommand{\KwPrec}{\mathrm{prec}}
\newcommand{\KwLast}{\mathrm{last}}
\newcommand{\KwVarB}{\mathrm{S}}
\newcommand{\seco}{;\hskip 0.25cm }
\newcommand{\KwGetBordMax}{\texttt{GetLastBord}}
\begin{algorithm}[htp]
\SetKw{KwAnd}{and}
\SetKw{KwCompute}{compute}
\SetKw{print}{print}
\SetKwInOut{Input}{input}
\SetKwInOut{Output}{output}
	\SetAlgoLined%\DontPrintSemicolon{}
	\LinesNumbered
	\LinesNumbered
%	\Input{a pattern $w$}
%	\Output{the position lattice of $w$}
%	\vskip 0.2cm	
{\small
		$\KwBordMax\leftarrow$ length of the longest proper suffix of $w$ which is also a prefix\;
		\lFor{$x\in\alp$} {
			$\KwLast[x]\leftarrow\nul$
		}
		\For{$i = 0$ \KwTo $\size{w}-1$\nllabel{liLoopSA}} {
			$\KwLast[w_{i}]\leftarrow i$\;
			\For{$x\in\alp$\nllabel{liLoopSAI}} {
				\If{$\KwLast[x]\neq\nul$}{
					$\shiftLattice{\emptyset}(i,x)\leftarrow i-\KwLast[x]$\seco $\transLattice{\emptyset}(i,x)\leftarrow\{\KwLast[x]\}$\;
				} \Else {
					$\shiftLattice{\emptyset}(i,x)\leftarrow i+1$\seco $\transLattice{\emptyset}(i,x)\leftarrow\emptyset$\nllabel{liLoopEA}\;
				}
			}
		}
		\For{$i = 0$ \KwTo $\size{w}-1$\nllabel{liLoopSB}} {
			\For{$j = 0$ \KwTo $i-1$} {
				\For{$x\in\alp$} {
						$\shiftLattice{\{i\}}(j,x)\leftarrow \shiftLattice{\emptyset}(j,x)+\shiftLattice{\transLattice{\emptyset}(j,x)}(i-\shiftLattice{\emptyset}(j,x), w_{i})$\nllabel{liLoopSBLa}\;
						$\transLattice{\{i\}}(j,x)\leftarrow \transLattice{\transLattice{\emptyset}(j,x)}(i-\shiftLattice{\emptyset}(j,x), w_{i})$\nllabel{liLoopSBLb}\;
				}
			}
			\For{$j = i+1$ \KwTo $\size{w}-1$} {
				\For{$x\in\alp$} {
					\If{$x=w_{j}$}{
						\If{$\size{w}=2$} {
							$\shiftLattice{\stl}(i,x) \leftarrow \size{w}-\KwBordMax$\seco $\transLattice{\stl}(i,x) \leftarrow \{0,\ldots,\KwBordMax-1\}$\;
						} \Else {
                            	$\shiftLattice{\{i\}}(j,x)\leftarrow 0$\seco $\transLattice{\{i\}}(j,x)\leftarrow \{i,j\}$\;
						}
					} \Else {
						$\shiftLattice{\{i\}}(j,x)\leftarrow \shiftLattice{\transLattice{\emptyset}(i-1,w_{i})}(j-\shiftLattice{\emptyset}(i-1,w_{i}),x)$\nllabel{liLoopSBLc}\;
						$\transLattice{\{i\}}(j,x)\leftarrow \transLattice{\transLattice{\emptyset}(i-1,w_{i})}(j-\shiftLattice{\emptyset}(i-1,w_{i}),x)$\nllabel{liLoopEB}\;
					}
				}
			}
		}
		\For{$\ell = 2$ \KwTo $\size{w}-1$\nllabel{liLoopSC}} {
			\lFor{$j = 0$ \KwTo $\ell-1$} {
				$\KwVarB[j]\leftarrow j$
			}
			\Repeat{$j<0$}{
				$\stl\leftarrow \{\KwVarB[0], \ldots, \KwVarB[\ell-1]\}$\seco $\stl'\leftarrow \{\KwVarB[0], \ldots, \KwVarB[\ell-2]\}$\;
				\For{$i\in\setMinus{\{0,\ldots,\size{w}-1\}}{s}$} {
					\For{$x\in\alp$} {
                    	\If{$x=w_{i}$} {
                    		\If{$\ell=\size{w}-1$} {
								$\shiftLattice{\stl}(i,x) \leftarrow \size{w}-\KwBordMax$\seco $\transLattice{\stl}(i,x) \leftarrow \{0,\ldots,\KwBordMax-1\}$\;
							} \Else {
                            	$\shiftLattice{\stl}(i,x)\leftarrow 0$\seco $\transLattice{\stl}(i,x)\leftarrow \stl\cup\{i\}$\;
							}
						} \Else {
							$\shiftLattice{\stl}(i,x) \leftarrow \shiftLattice{\stl'}(i,x) + \shiftLattice{\transLattice{\stl'}(i,x)}(\KwVarB[\ell-1]-\shiftLattice{\stl'}(i,x), w_{\KwVarB[\ell-1]})$\;
						$\transLattice{\stl}(i,x) \leftarrow \transLattice{\transLattice{\stl'}(i,x)}(\KwVarB[\ell-1]-\shiftLattice{\stl'}(i,x), w_{\KwVarB[\ell-1]})$\;
						}
					}
				}
				$j\leftarrow\ell-1$\nllabel{liLoopSCLa}\;
				\lWhile{$j\geq 0$ \KwAnd $\KwVarB[j]\geq \size{w}-\ell+j$} {
					$j\leftarrow j-1$
				}
				\If{$j\geq 0$} {
					$\KwVarB[j]\leftarrow \KwVarB[j]+1$\;
					\lFor{$k = j+1$ \KwTo $\ell-1$} {
						$\KwVarB[k]\leftarrow \KwVarB[k-1]+1$
					}
				}
			}\nllabel{liLoopEC}
		}
%		\vskip 0.3cm
}
\caption{Computation of the position lattice. Value $\KwBordMax$ is the last entry of the \textit{Partial Match table} of the KMP algorithm. Its computation takes a time linear with $\size{w}$ \cite{Knuth1977}.}\label{algoLatt}
\end{algorithm}

\begin{theorem}
	Algorithm \ref{algoLatt} computes the  position lattice of the pattern $w$ in $\order{\size{w}2^{\size{w}}}$ time by using the same amount of memory.
\end{theorem}
\begin{proof}
	Let us first show that Algorithm \ref{algoLatt} determines the shifts and the transitions of the state $\stl$ before those of the state $\stl'$ if and only if $\stl\lsl\stl'$. The loop at Lines \ref{liLoopSA}-\ref{liLoopEA} computes the shifts and the transitions of $\emptyset$. Next, the loop at Lines \ref{liLoopSB}-\ref{liLoopEB} computes the shifts and the transitions of the singletons from $\{0\}$ to $\{\size{w}-1\}$. The last loop (Lines \ref{liLoopSC}-\ref{liLoopEC}) determines the shifts and the transitions of the states corresponding to the subsets of increasing cardinals  $\ell$ from $2$ to $\size{w}-1$. Inside the last loop, the way in which the next subset $\stl'$ is computed from the current subset $\stl$, both of cardinal $\ell$, ensures that $\stl\lsl\stl'$  (Lines \ref{liLoopSCLa}-\ref{liLoopEC}).
	
	For all iterations $i$ of the loop at Lines \ref{liLoopSA}-\ref{liLoopEA} and all symbols $x$, we have $\KwLast[x]=\precocc{i}{x}$ at the beginning of the inner loop (Line \ref{liLoopSAI}). From Lemma \ref{lemmaLattice} (Cases 1 and 2a), the transitions $\transLattice{\emptyset}(i,x)$ and the shifts $\shiftLattice{\emptyset}(i,x)$ for all positions $i$ of $w$ and all symbols $x$, are correctly computed at the end of the loop.
	
	The loop at Lines \ref{liLoopSB}-\ref{liLoopEB} computes the shifts and the transitions from the singleton states. For all pairs of positions $(i,j)$ and all symbols $x$, determining $\transLattice{\{ i \}}(j,x)$ and $\shiftLattice{\{ i \}}(j,x)$ is performed by distinguishing between two cases.
	\begin{itemize}
	\item If $i>j$, then $\transLattice{\emptyset}(j,x)\lsl\{i\}$ and its shifts and transitions were already computed. Formula of Remark \ref{remComp} gives us those of $\{ i\}$ (Lines \ref{liLoopSBLa}-\ref{liLoopSBLb}). 
	\item If $i<j$, we distinguish between two subcases according to the symbol $x$ considered. If $x=w_{j}$ then the shift and the transition state are given in Lemma \ref{lemmaLattice} - Case 1. Otherwise, we remark that, since $\shiftLattice{\{ i \}}(j,x)$ is positive, we have that $\shiftLattice{\{ i \}}(j,x) = \min\{k\geq 1\condi  w_{j-k} = x \mbox{ if } j\geq k \mbox{ and } w_{i-k} = w_{i}\}$. This implies that $\shiftLattice{\{ i \}}(j,x) = \shiftLattice{\transLattice{\emptyset}(i-1,w_{i})}(j,x)$. We have $\transLattice{\emptyset}(i-1,w_{i})\lsl\stl$, thus both the shifts and the transitions of the state $\transLattice{\emptyset}(i-1,w_{i})$ are computed before $\stl$ (Lines \ref{liLoopSBLc}-\ref{liLoopEB}).
	\end{itemize}
	
	The last loop, lines \ref{liLoopSC}-\ref{liLoopEC}, computes the shifts and the transitions of the states corresponding to the subsets of cardinals $2$ to $\ell-1$. For all states $\stl$ with $2\leq\card{\stl}\leq\size{w}-1$, all positions $i\in\complementary{\stl}$ and all symbols $x\neq w_{i}$, the corresponding shift and transition $\shiftLattice{\stl}(i,x)$ and $\transLattice{\stl}(i,x)$ are computed from the shifts and transitions of the state $\transLattice{\setMinus{\stl}{{\{\max\stl\}}}}(i,x)$ following Lemma \ref{lemmaLattice} - Cases 2b (in Algorithm \ref{algoLatt}, we put $\stl'$ for $\setMinus{\stl}{{\{\max\stl\}}}$). Lemma \ref{lemmaOrder} ensures that $\transLattice{\setMinus{\stl}{{\{\max\stl\}}}}(i,x)\lsl\stl$, thus that the shifts and transitions of $\transLattice{\setMinus{\stl}{{\{\max\stl\}}}}(i,x)$ are computed before those of $\stl$. For all states $\stl$ with $2\leq\card{\stl}\leq\size{w}-1$, all positions $i\in\complementary{\stl}$, the shift and transition $\shiftLattice{\stl}(i,w_{i})$ and $\transLattice{\stl}(i,w_{i})$  are given in Lemma \ref{lemmaLattice} - Case 1.
	
	The time complexity is $\order{\sum_{k=0}^{\size{w}-1} (k+\size{w}-k) {\binom{\size{w}}{k}}}$ (loop Lines \ref{liLoopSC}-\ref{liLoopEC}), i.e. $\order{\size{w}2^{\size{w}}}$. We do not use more memory than needed to store the lattice, which is, from Remark \ref{remSizeLattice}, $\order{\size{w}2^{\size{w}}}$.
\end{proof}

\section{The Fastest $w$-strategy}\label{secBrute}

Determining the fastest $w$-strategy, which, from Proposition \ref{propOptim}, has the greatest asymptotic speed among all the $w$-matching machines of order $\size{w}-1$, may be performed by computing the asymptotic speed of all the $w$-strategies and by returning the fastest one.

In order to enumerate all the $w$-strategies, let us remark that they are all contained in the position lattice of $w$ in the sense that:
\begin{itemize}
\item the set of states of a $w$-strategy is included in that of the position lattice;
\item all the $w$-strategies $\matchine=(\stateset, \init, \prematchset, \nextmm, \trans, \shift)$ verify $\trans(\stl,x) = \transLattice{\stl}(\nextmm(\stl),x)$ and $\shift(\stl,x) = \shiftLattice{\stl}(\nextmm(\stl),x)$ for all $\stl\in\stateset$ and all symbols $x$.
\end{itemize}

Reciprocally, to any map $\mapnext$ from $\stateLattice$ to $\{0,\ldots, \size{w}-1\}$ such that $\mapnext(\stl)\in\complementary{s}$ for all states $\stl\in\stateLattice$,  there corresponds the unique $w$-strategy $\strat=(\stateset, \init, \prematchset, \nextmm, \trans, \shift)$ for which the next-position-to-check function $\nextmm$ coincides with $\mapnext$ on $\stateset$.

Finally, our brute force algorithm 
\begin{enumerate}
\item takes as input a pattern $w$ and an iid model $\piid$,
\item computes the position lattice of $w$,
\item enumerates all the maps $\mapnext$ such that $\mapnext(\stl)\in\complementary{s}$ for all states $\stl\in\stateLattice$,
\item for each $\mapnext$, gets the corresponding $w$-strategy by keeping only the states of $\stateLattice$ reachable from $\emptyset$, with the next-position-to-check function $\mapnext$,
\item computes the asymptotic speed of all the $w$-strategies under $\piid$,
\item returns the $w$-strategy with the greatest speed.
\end{enumerate}

The time complexity of the brute force algorithm is 
\begin{dmath*}
%O\left(\left(\prod_{k=1}^{\size{w}-1} k^{\binom{\size{w}}{k}}\right)2^{3\size{w}}\right),
O\left(\left(\prod_{k=1}^{\size{w}-1} (\size{w}-k)^{\binom{\size{w}}{k}}\right)2^{3\size{w}}\right),
\end{dmath*}
where the first factor stands for the number of functions $\mapnext$ and the second one for the computation of the asymptotic speed of a $w$-strategy, which needs to solve a linear system of size equal to the number of states, which is $O(2^{\size{w}})$. Its memory space complexity is  $\size{w}2^{\size{w}-1}$, i.e. what is needed to store the position lattice of $w$.

Under its current implementation, the brute force determination of the fastest $w$-strategy is unfeasible for patterns of length greater than $4$.

\section{A polynomial heuristic}\label{secHeuristic}
There are two points which make the complexity of the brute force algorithm given in Section {\ref{secBrute}} that high:
\begin{enumerate}
	\item the size of the position lattice, which is exponential with the length of the pattern,
	\item determining the fastest strategy in the position lattice, which needs a time exponential with its size. 
\end{enumerate}

Our heuristic is based on two independent stages, each one aiming to overcome one of these two points. Both of them start from the general idea that, since, for any current position of the text, the probability that no mismatch occurs until the $n^{\mbox{\tiny th}}$ text access decreases geometrically with $n$, the first relative positions accessed by a strategy (or more generally by a pattern algorithm) are those which have the greatest influence on its asymptotic speed.

\newcommand{\prefset}[1]{P(#1)}
\newcommand{\remaset}[1]{R(#1)}
\newcommand{\setsub}[0]{\set{U}}

\subsection{$n$-sets sublattices}

A sufficient condition for a sublattice $\setsub\subseteq\stateLattice$ to contain a $w$-strategy is that, for all $\stl\in\setsub$, there exists at least a position $i\in\complementary{\stl}$ with $\transLattice{\stl}(i,x)\in\setsub$ for all $x\in\alp$. A sublattice $\setsub$ verifying this condition will be said to be $\defi{complete}$. Figure \ref{figSublattice} displays four complete sublattices extracted from the position lattice of $abb$ (Figure \ref{figLattice}).

\begin{figure*}[!tpb]
\centering{\includegraphics[width=\textwidth, trim=0cm -0.0cm 0cm 0cm]{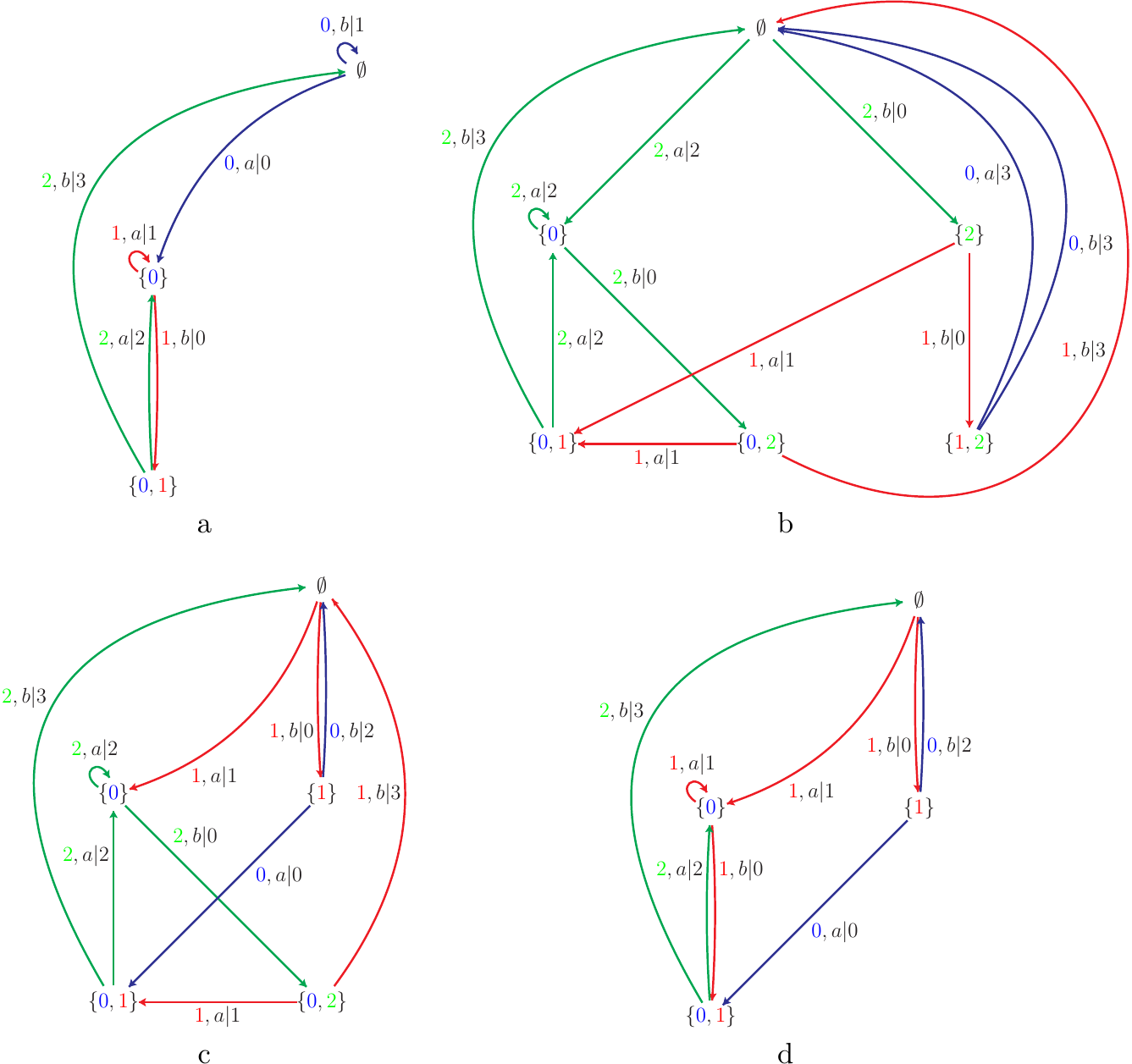}}
\caption{Four complete sublattices extracted from the position lattice of $abb$. Sublattice a (resp. b) leads to the strategy where the next-position-to-check is always the smallest (resp. the greatest) relative position unchecked.
Sublattice c (resp. d) leads to the strategy at the top (resp. at the bottom) of Figure \ref{figStrategies}.}
\label{figSublattice}
\end{figure*}

Let us introduce some additional notations here. For all sets $\set{S}$ of positions,  the \defi{prefix} of $\set{S}$ is defined as $\prefset{\set{S}} = \max\{i\in\set{S}\condi  j\in\set{S}\mbox{ for all } 0\leq j\leq i\}$ and its \defi{rest} is $\remaset{\set{S}} = \setMinus{\set{S}}{\{0, \ldots, \prefset{\set{S}}\}}$

For all positive integers $n$, the $n$-sets sublattice of $w$ is the sublattice $\setsub$ of $\stateLattice$ which contains all and only the subsets of $\stateLattice$ with a rest containing less than $n$ positions, i.e. the subsets of the form $\{0,\ldots,p\}\cup \set{X}$ with $p<\size{w}-1$ and $\card{\set{X}}\leq n$. 

By construction, the $n$-sets sublattice of $w$ is complete. It contains $O(\size{w}^{n})$ states and $O(\size{w}^{n+1})$ transitions.

We adapted Algorithm \ref{algoLatt} to compute the $n$-sets sublattice of $w$ in $O(\size{w}^{n+1})$ time with the same amount of memory space.

\subsection{$\ell$-shift expectation}
We are now interested in a fast way for finding an efficient $w$-strategy in a given complete sublattice.

\newcommand{\setTrans}[2][\setsub]{\mathrm{Tr}(#2)}

\newcommand{\ksexp}[3][w]{\mathrm{ES}_{#2}^{[#1]}[#3]}
For all integers $\ell$ and all states $\stl$ of a sublattice $\setsub$, the $\ell$-shift expectation of $\stl$ is defined as the greatest shift expectation one could possibly get in $\ell$ steps in $\setsub$ by starting from $\stl$, conditioned on starting from $\stl$, while parsing a text following an iid model $\piid$. Namely, the $\ell$-shift expectation is computed following the recursive formula:
%For all states $\stl$ of $\stateLattice$, the $k$-shift expectation $\ksexp{k}{\stl}$   is recursively defined as 
\begin{itemize}
	\item $\ksexp{0}{\stl} = 0$,
	\item for all $\ell>0$, $$\ksexp{\ell}{\stl} = \max_{i\in\setTrans{\stl}}\sum_{x\in\alp}\piid(x)\left(\shiftLattice{\stl}(i,x)+\ksexp{\ell-1}{\transLattice{\stl}(i,x)}\right)$$
\end{itemize}
where $\setTrans{\stl} = \{i\in\complementary{\stl}\condi \transLattice{\stl}(i,x)\in\setsub \mbox{ for all }x\in\alp\}$.

The $\ell$-shift expectation of a complete sublattice $\setsub$ is well defined and can be computed in $O(\ell T)$ time, where $T$ is the number of transitions of the sublattice $\setsub$ and by using $O(\card{\setsub})$ memory space.

%The $k$-shift expectation yields to 
We finally extract a $w$-strategy from $\setsub$ by setting the next-position-to-check of all states $\stl\in\setsub$ to
\begin{dmath*}
\argmax_{i\in\setTrans{\stl}}\sum_{x\in\alp}\left(\shiftLattice{\stl}(i,x)+\ksexp{\ell-1}{\transLattice{\stl}(i,x)}\right).
\end{dmath*}
% which is such that  $\ksexp{k}{\stl} = \sum_{x\in\alp}\left(\shiftLattice{\stl}(i,x)+\ksexp{k-1}{\transLattice{\stl}(i,x)}\right)$.

\subsection{$K$-Heuristic}

The $K$-Heuristic combines the two approaches above in order to compute a $w$-strategy in a time polynomial with the length of the pattern.  

Being given an order $K\geq 1$, we start by computing the $K$-sets sublattice of $w$, thus in $O(\size{w}^{K+1})$ time. In order to select a $w$-strategy from the $K$-sets sublattice, we next compute the $(K+1)$-shift expectation of all its states and extract a $w$-strategy as described just above. This computation is performed in $O(K \size{w}^{K+1})$ time, since the number of transition of the sublattice is $O(\size{w}^{K+1})$, by using $O(\size{w}^{K})$ memory space.

Let us remark that the order $\ell$ of the $\ell$-shift expectation does not have, \textit{a priori}, to be strongly related to the order $K$ of the $K$-sets sublattice on which it is computed. By experimenting various situations, we observed that considering an order greater than $K+1$ generally does not improve much the performances, whereas the strategies obtained from $\ell$-expectations with $\ell$ smaller than $K$ may be significantly slower.

The $K$-Heuristic returns a $w$-strategy in $O(K \size{w}^{K+1})$ time by using $O(\size{w}^{K})$ memory space. We insist on the fact that the $K$-Heuristic generally does not return the fastest strategy, even if $K>\size{w}$. However, we will see in the next section that it performs quite well in practice.

\section{Evaluation}\label{secEvaluation}
We shall compare the approaches introduced in Sections \ref{secBrute} and \ref{secHeuristic} with selected pattern matching algorithms. The comparison is performed, first, from a theoretical point of view, by computing their asymptotic speeds under iid models, and second, in practical situations, by measuring their average speed over real data. The \defi{average speed} with regard to a pattern $w$, of an algorithm or a matching machine on a text $t$ is the ratio of $\size{t}$ to the number of text accesses performed by the algorithm to search $w$ in $t$.

We are also interested in to what extent taking into account the frequencies of the letters of an iid model or a text, for determining the Fastest and the $K$-Heuristic strategies, actually improves their asymptotic or their average speeds. To this purpose, we compute the  Fastest and the $K$-Heuristic strategies from the uniform iid model. Next, we test their efficiency in terms of asymptotic speed under a non-uniform iid model and in terms of average speeds on data with non-uniform frequencies of letters.

\subsection{Pre-existing pattern matching algorithms}\label{secAlgoStd}
More than forty years of research have already led to the development of dozens algorithms. We selected the $\nAlgos$ ones below for our evaluation:
\begin{enumerate}
\item Naive \cite{Charras2004},
\item Morris-Pratt \cite{Charras2004},
\item Knuth-Morris-Pratt \cite{Knuth1977},
\item Quicksearch \cite{Sunday1990},
\item Boyer-Moore-Horspool \cite{Horspool1980},
\item TVSBS \cite{Thathoo2006}, a``right-to-left'' algorithm in which shifts are given by a bad-character rule  \cite{Boyer1977,Sunday1990} taking into account the two letters at distances $\size{w}-1$ and $\size{w}$ from the current position of the text,
\item EBOM \cite{Faro2009}, a version of the Backward Oracle Matching algorithm \cite{Allauzen2001} which also uses a ``bad two-characters'' rule,
\item HASHq \cite{Wu1994}, which implements the Boyer-Moore algorithm on blocks of length $q$ by using efficient hashing techniques \cite{Karp1987} (our tests are performed with $q=3$),
\item FJS \cite{Franek2005}, which combines the ideas of Knuth-Morris-Pratt \cite{Knuth1977} and Sunday \cite{Sunday1990} algorithms.
\end{enumerate}

Algorithms 1 to 5 are classics. The last four ones were chosen for being known to be efficient on short patterns and small alphabets \cite{Faro2013}, a situation in which the determination of the fastest strategy is feasible.

Let us remark that the order of the $w$-matching machine associated to TVSBS is equal to $\size{w}$, thus greater than that of the Fastest strategy that we compute.

The transformation into matching machines was implemented for a few other pattern matching approaches, for instance the SA  algorithm (the Baeza-Yates-Gonnet algorithm) based on bitwise operations \cite{baeza1992new}, or the string-matching automaton  \cite{Cormen1990}. Since the asymptotic and average speeds of these two algorithms are exactly $1$, whatever the pattern, the model and the text, there is no point in displaying them.

\subsection{Results}
We shall evaluate:
\begin{itemize}
\item the pre-existing pattern matching algorithms presented in Section \ref{secAlgoStd},
\item the $1$- $2$- and $3$-Heuristics and
\item the Fastest strategy (each time it is possible).
\end{itemize}

\subsubsection{Asymptotic speed}
The asymptotic speeds are computed for texts and patterns on the binary alphabet $\{\mathtt{a},\mathtt{b}\}$.

\begin{table*}[!tpb]
\centering{
\includegraphics[width=\textwidth]{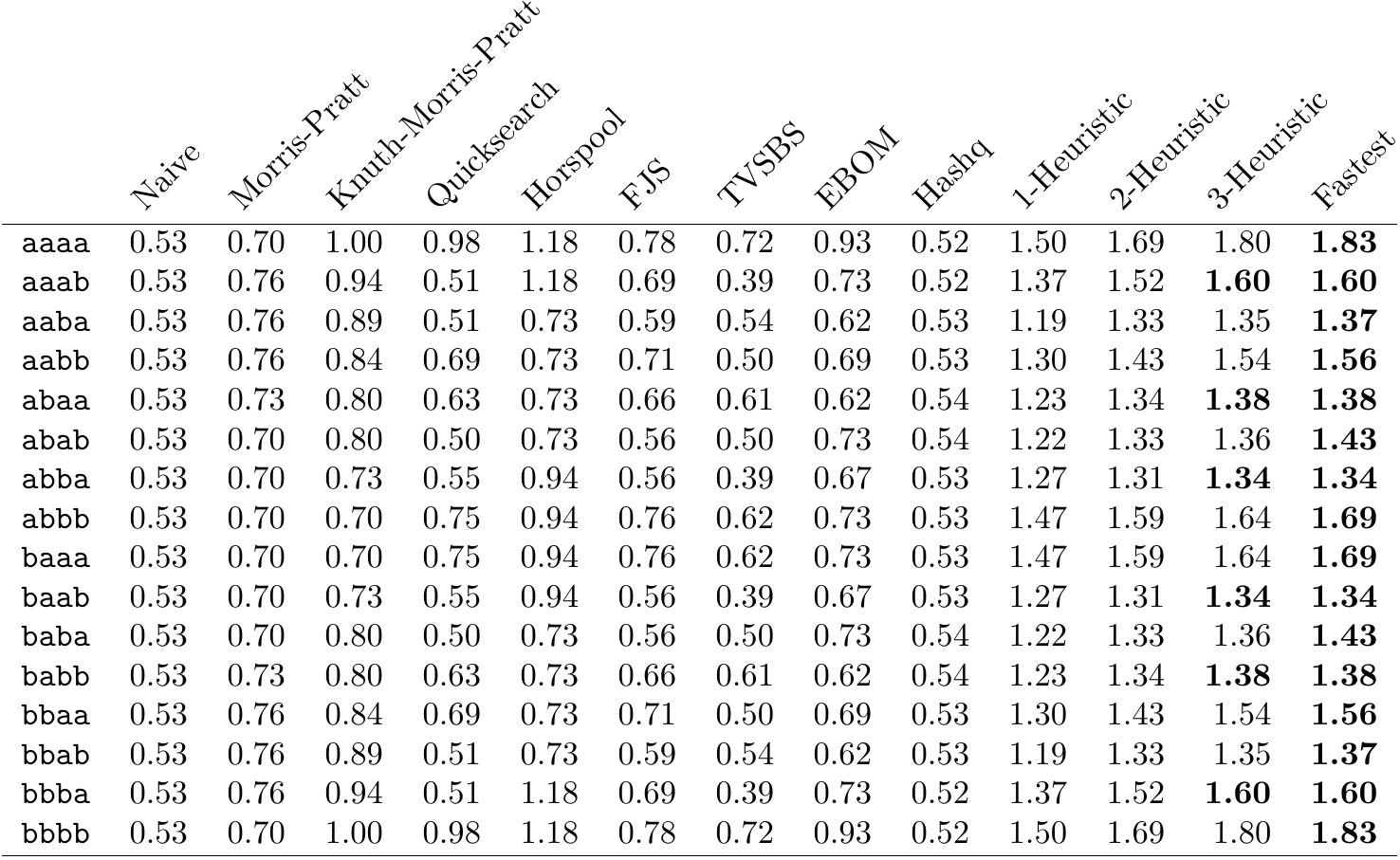}
}
\caption{Asymptotic speeds for the patterns of length $4$ on $\{\mathtt{a},\mathtt{b}\}$ under the uniform model.}
\label{tabTheo5}
\end{table*}
Table \ref{tabTheo5} displays the asymptotic speeds for all the patterns of length $4$ on iid texts drawn from the uniform distribution.
As expected, the strategy computed with the brute force algorithm (last column) is actually the fastest, but the speeds of the  $1$-,$2$- and $3$-Heuristics are very close. The pre-existing algorithms are outperformed by all our approaches (even by the $1$-Heuristic) for all the patterns. 
We observe that the Naive, Morris-Pratt and Knuth-Morris-Pratt algorithms have asymptotic speeds always smaller than $1$. One cannot expect them to be faster since, by construction, they access all the positions of a text at least once. In the following, we will not display their speeds, nor that of Quicksearch, for they are always smaller than at least one of the other pre-existing algorithms. 
The full tables can easily be re-computed by using our software.

\begin{table*}[!tpb]
\centering{
\includegraphics[width=\textwidth]{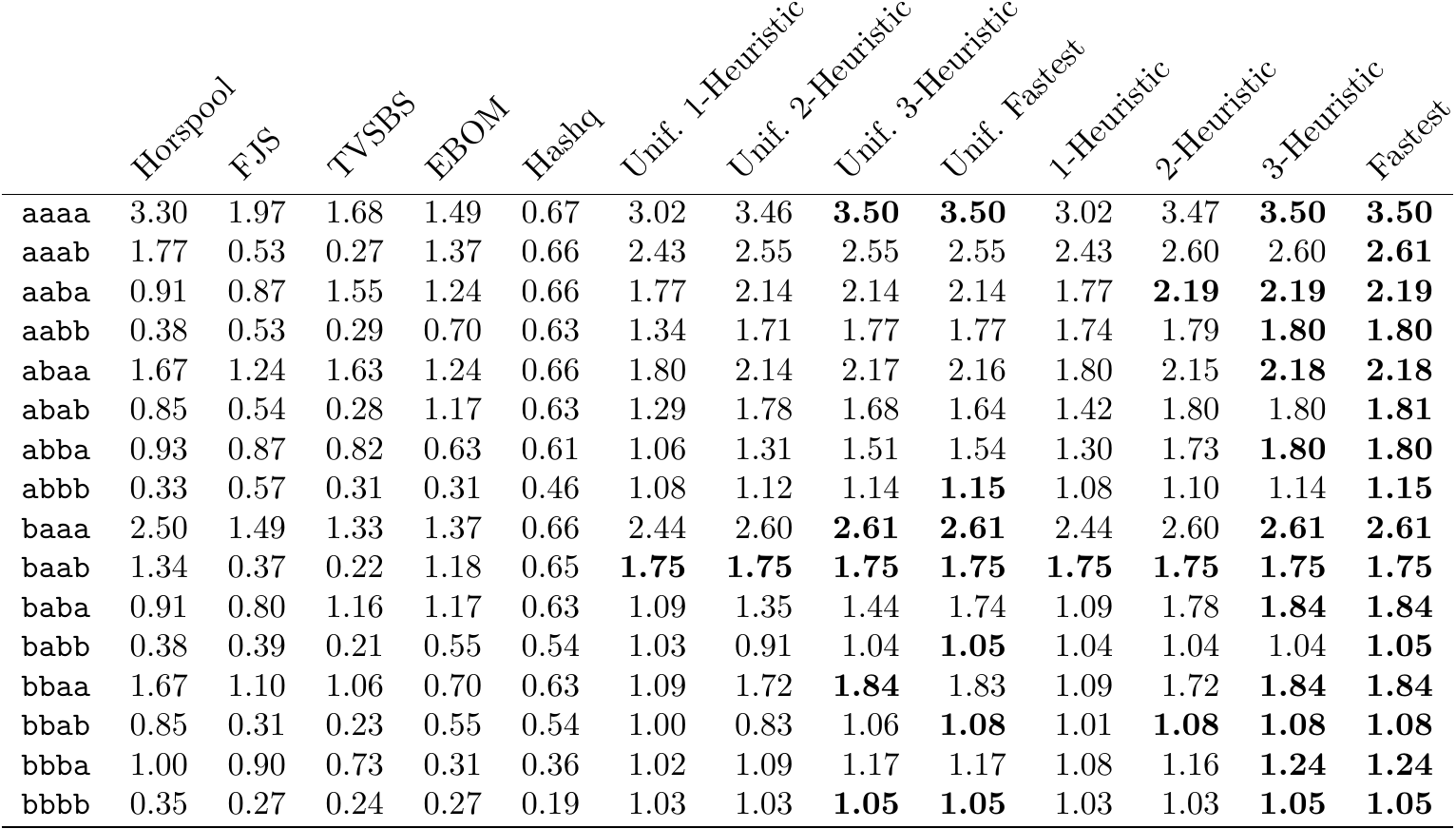}
}
\caption{Asymptotic speeds for the patterns of length $4$ on $\{\mathtt{a},\mathtt{b}\}$ under the iid model $(\piid_{\mathtt{a}},\piid_{\mathtt{b}}) = (0.1, 0.9)$.}
\label{tabTheo1}
\end{table*}

Table \ref{tabTheo1} displays the asymptotic speeds with regard to the same patterns as Table \ref{tabTheo5}, but under the iid model $(\piid_{\mathtt{a}},\piid_{\mathtt{b}}) = (0.1, 0.9)$. This table shows the asymptotic speeds of the $K$-Heuristics and the Fastest strategies computed with regard to an uniform iid model (the columns starting with ``Unif.''). The strategies such obtained are not optimized according to the letter probabilities of the model. They may be used as general purpose approaches, while the strategies obtained from the model probabilities will be called \defi{adapted} below.
Overall, our methods are faster than the pre-existing algorithms, with a few exceptions: Horspool is faster than the $1$-Heuristic for two patterns ending with the rare letter $\mathtt{a}$: $\mathtt{aaaa}$ and $\mathtt{baaa}$. And EBOM is faster than the $1$-Heuristic for searching $\mathtt{baba}$.
The $K$-Heuristics and the Fastest strategies computed with regard to an uniform iid model have asymptotic speeds smaller than their counterparts obtained from the actual probabilities of the text model (here highly unbalanced). Nevertheless, the uniform approaches still perform quite well, notably better than the pre-existing algorithms, except for the uniform $1$-Heuristic and the same patterns as above.

\begin{table*}[!tpb]
\centering{
\includegraphics[width=\textwidth]{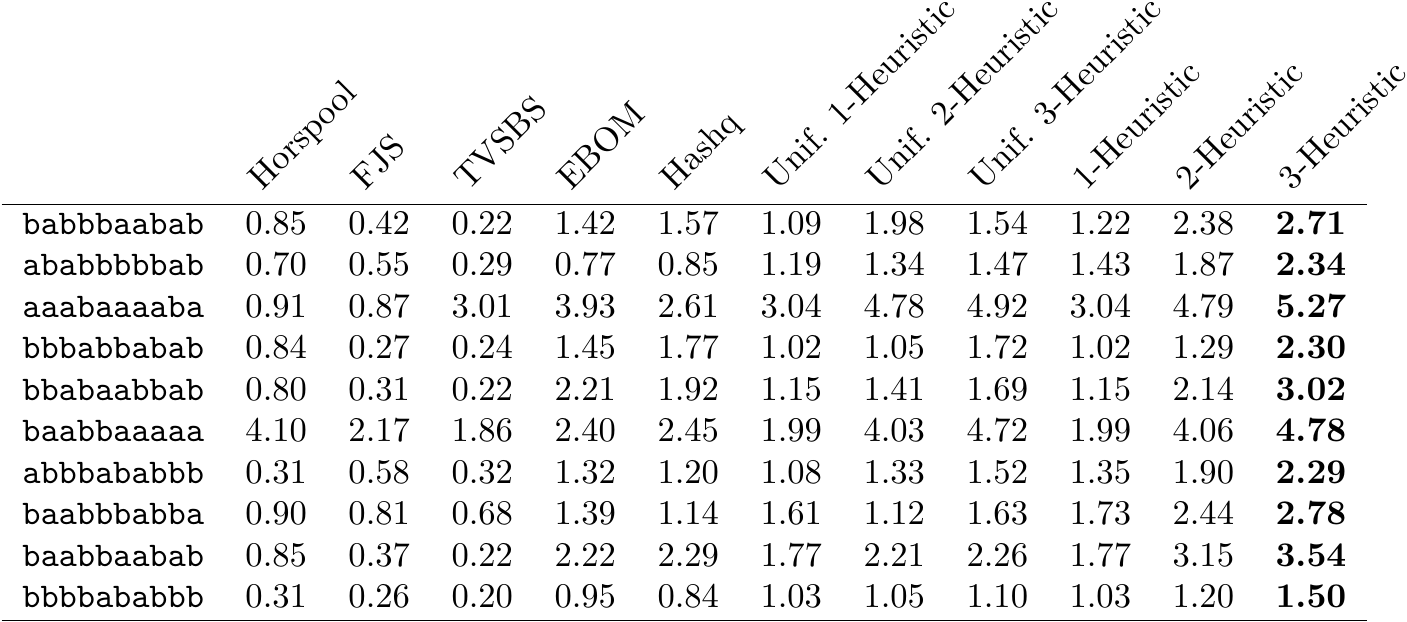}
}
\caption{Asymptotic speeds for some patterns of length $10$ on $\{\mathtt{a},\mathtt{b}\}$ (drawn from the uniform distribution) under the iid model $(\piid_{\mathtt{a}},\piid_{\mathtt{b}}) = (0.1, 0.9)$.}
\label{tabTheo10}
\end{table*}

Considering longer patterns leads to similar observations. Table \ref{tabTheo10} shows the asymptotic speeds obtained for random patterns of length $10$. The $3$-Heuristic outperforms all the others approaches (the Fastest strategy cannot be computed for this length). The (uniform) $1$-Heuristic is slower than algorithms such EBOM or Hashq. But both the uniform $2$- and $3$-Heuristic  overall perform better than the pre-existing algorithms, though they are slightly slower for a few patterns.

\subsection{Average speed}
Our data benchmark consists in the \textit{Wigglesworthia glossinidia} genome, known for its bias in nucleotide composition ($78\%$ of $\{\mathtt{a,t}\}$), and the Bible in English from \cite{Faro2013}.
\begin{table*}[!tpb]
\centering{
\includegraphics[width=\textwidth]{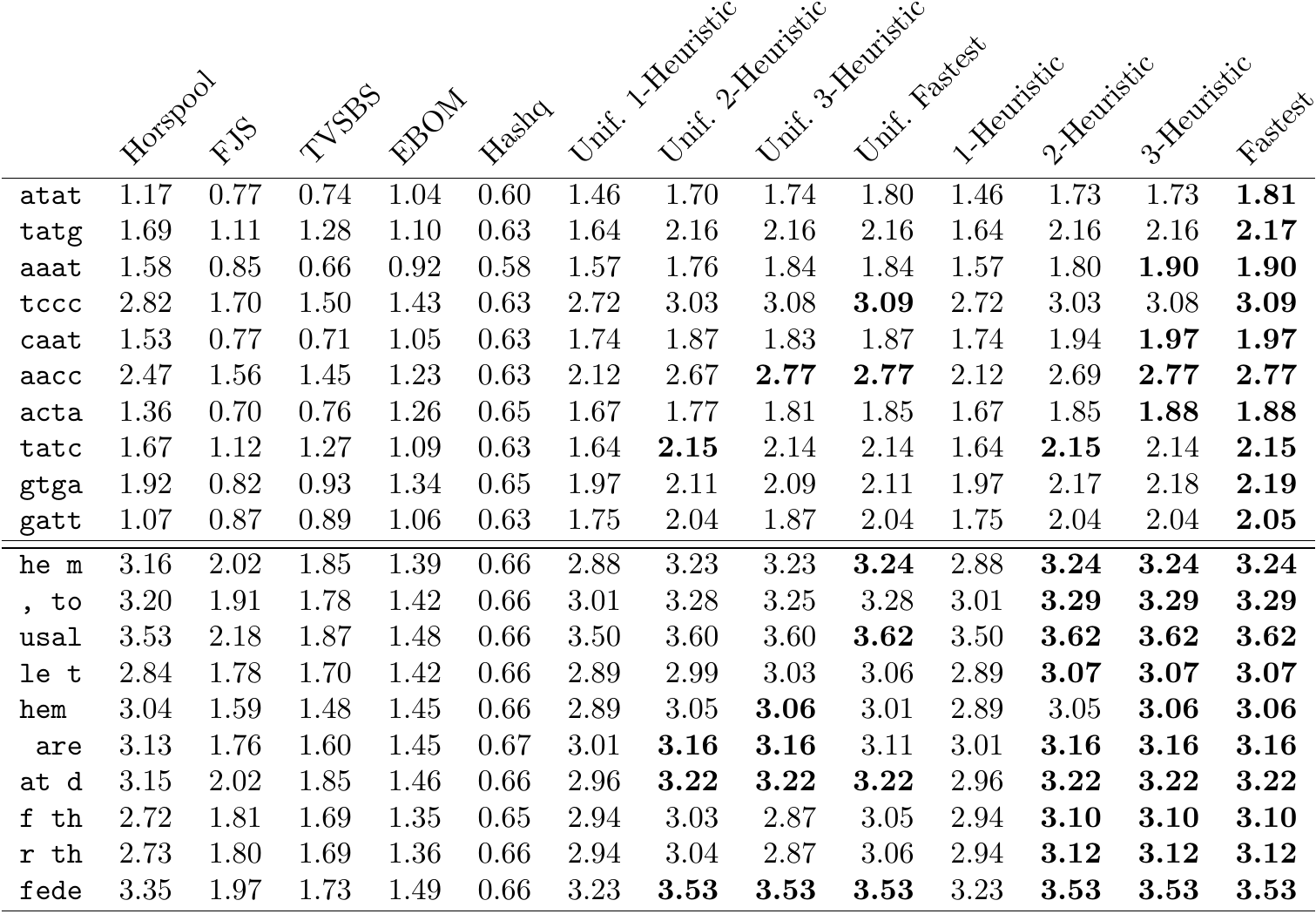}
}
\caption{Average speeds for some patterns of length $4$ picked from the benchmark data (the \textit{Wigglesworthia glossinidia} complete genome and the Bible in English).}
\label{tabPracti4}
\end{table*}

Table \ref{tabPracti4} displays the average speeds of patterns randomly picked from the data. Let us remark that we are now dealing with real texts, which are not iid. In particular, the Fastest strategy could possibly be outperformed (this is not observed on the benchmark data). The $2$- and $3$-Heuristics, uniform and adapted, are faster than the pre-existing algorithms for all the patterns, whereas the $1$-Heuristic is sometimes slightly outperformed by Horspool. Horspool is almost as fast as our approaches on the Bible while being sometimes significantly outperformed on the \textit{Wigglesworthia glossinidia} genome. The average speeds are overall greater on the Bible than on the DNA sequence. In both cases, we do not observe a wide performance gap between the uniform and the adapted approaches, though our benchmark data are far from following an uniform iid model.
Let us remark that the $2$- and $3$-Heuristics have almost the same performances both in the uniform and the adapted cases.
\begin{table*}[!tpb]
\centering{
\includegraphics[width=1.03\textwidth]{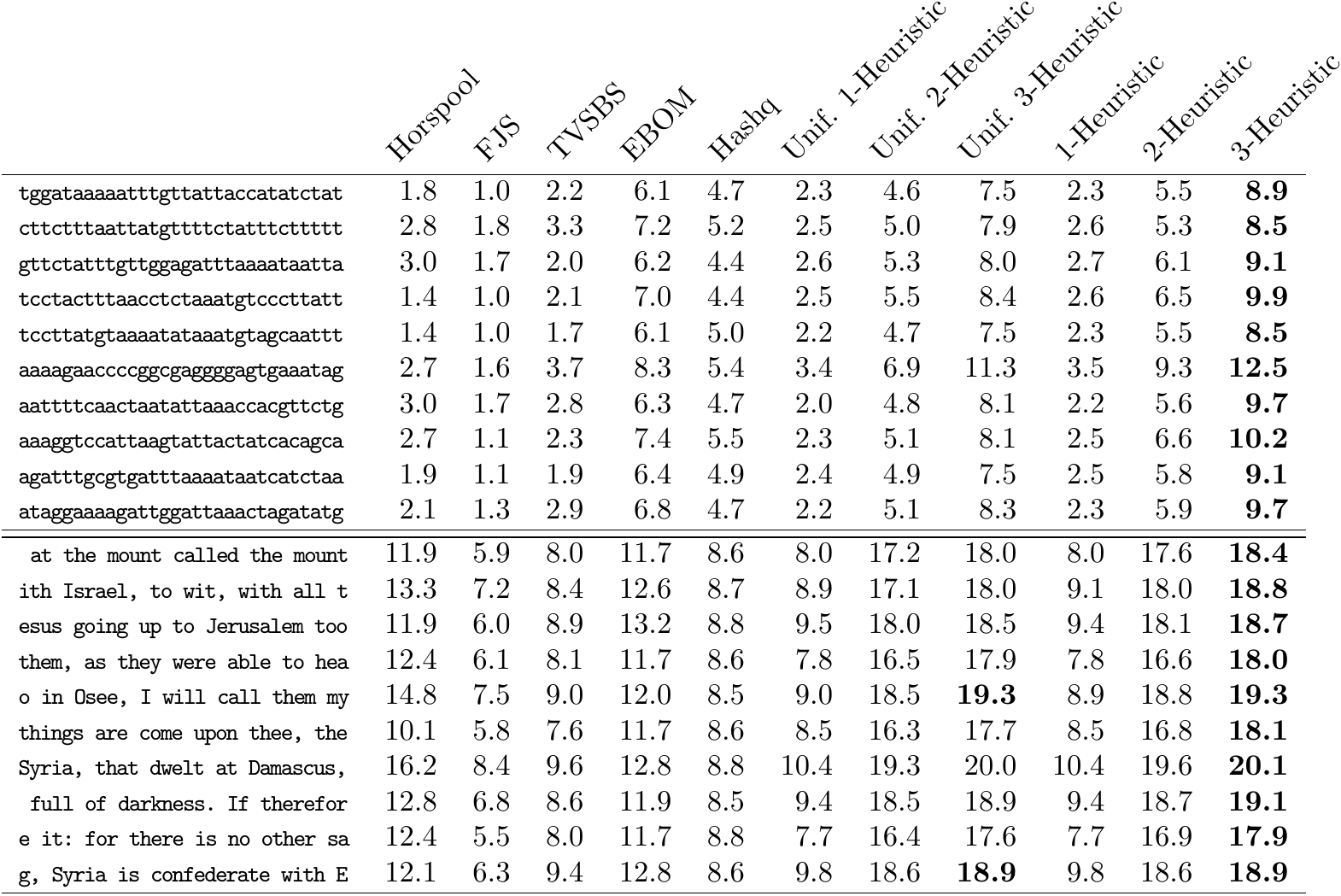}
}
\caption{Average speeds for some patterns of length $30$ picked from the benchmark data (the \textit{Wigglesworthia glossinidia} complete genome and the Bible in English).}
\label{tabPracti30}
\end{table*}
 
Table \ref{tabPracti30} shows the averages speeds with regard to patterns of length $30$. The average speeds on the Bible are about twice those on the \textit{Wigglesworthia glossinidia} genome. One actually expects the speed to be greater in average on texts with large alphabets, since the less likely the match between two symbols, the greater the shift expectation per iteration. Again the $3$-Heuristic, uniform or adapted, outperforms the pre-existing algorithms.
The speeds of the $3$-Heuristic and of the $2$-Heuristic differ in a greater amount than with patterns of length $4$ for the  \textit{Wigglesworthia glossinidia} genome, and, to a smaller extent, for the Bible.

\section{Discussion}
In practical situations and though they do not take into account the letter frequencies, the uniform $K$-Heuristics and the uniform Fastest strategy perform generally almost as well as their adapted counterparts. The greatest difference observed is for the patterns of length $30$ on the \textit{Wigglesworthia glossinidia} genome (Table \ref{tabPracti30}) and is relatively small.
We do observe a notable amount of difference for the quite extreme case of the asymptotic speed under the iid model $(\piid_{\mathtt{a}},\piid_{\mathtt{b}}) = (0.1, 0.9)$. But even for these frequencies, the uniform approaches show greater asymptotic speeds than any of the selected pre-existing algorithms.

The $3$-Heuristic has very good results whatever the pattern or the text. There is no situation for which the performances of the $2$-heuristic are far from the best. On the contrary, the performance ranking of the pre-existing algorithms depends heavily on the patterns and on the texts or the model. For instance, Horspool may perform very well, even almost optimally, for some patterns and texts or models while its speed may completely plummet in other situations.

The question of selecting the most efficient order of $K$-Heuristic still deserves further investigations. A basic answer could be ``the greater, the better'' but we should take into consideration that an higher order of heuristic comes with an increased computational cost.
After some experiments, we observed that the asymptotic speed the $K$-Heuristic tends to stop improving beyond a certain rank. For instance, the difference in average speed between the $2$- and $3$-Heuristics for patterns of length $4$, both on the genome and on the Bible, probably does not justify the computational cost of the $3$-Heuristic, while it is worth to use the $3$-Heuristic rather than the $2$-Heuristic for searching patterns of length $30$ in the Bible (not that much for the \textit{Wigglesworthia glossinidia} genome). The best trade-off for the order of the $K$-Heuristic depends on the pattern (notably its  length) and on the text features (in particular the alphabet size and the letter frequencies). 

It is certainly possible to obtain efficient heuristic with a lower computational cost than for the $K$-Heuristic. Since in standard situation, the length of the text is much greater than that of the pattern, there is no real reason for considering only pattern matching algorithms with linear pre-processings of the pattern. In the extreme case where the texts are arbitrarily long with regard to the patterns, any pre-processing, i.e. whatever its computation time, would be beneficial as soon as it improves the overall speed. 
\section*{Authors' contributions}
Gilles Didier provided the initial idea, led the software development and wrote all the manuscript but the section \textit{Evaluation}. Laurent Tichit collaborated on the software development, ran the tests and wrote the section \textit{Evaluation}. Both authors read, edited and approved the final manuscript.


\begin{thebibliography}{10}

\bibitem{Allauzen2001}
C.~Allauzen, M.~Crochemore, and M.~Raffinot.
\newblock Efficient experimental string matching by weak factor recognition.
\newblock In {\em Combinatorial Pattern Matching}, pages 51--72. Springer,
  2001.

\bibitem{baeza1992new}
R.~Baeza-Yates and G.~H. Gonnet.
\newblock A new approach to text searching.
\newblock {\em Communications of the ACM}, 35(10):74--82, 1992.

\bibitem{BaezaYates1992}
R.~A. Baeza-Yates and M.~R\'egnier.
\newblock Average running time of the {B}oyer-{M}oore-{H}orspool algorithm.
\newblock {\em Theoretical Computer Science}, 92(1):19 -- 31, 1992.

\bibitem{Barth1984}
G.~Barth.
\newblock An analytical comparison of two string searching algorithms.
\newblock {\em Information Processing Letters}, 18(5):249 -- 256, 1984.

\bibitem{Boyer1977}
R.~S. Boyer and J.~S. Moore.
\newblock A fast string searching algorithm.
\newblock {\em Communications of the ACM}, 20(10):762--772, 1977.

\bibitem{Charras2004}
C.~Charras and T.~Lecroq.
\newblock {\em Handbook of Exact String Matching Algorithms}.
\newblock King's College Publications, 2004.

\bibitem{Cormen1990}
T.~Cormen, C.~Leiserson, and R.~Rivest.
\newblock {\em Introduction to {A}lgorithms}.
\newblock MIT Press, 1990.

\bibitem{DidierX}
G.~Didier.
\newblock Optimal pattern matching algorithms.
\newblock \url{http://arxiv.org/abs/1604.08437}, 2016.

\bibitem{Faro2009}
S.~Faro and T.~Lecroq.
\newblock Efficient variants of the {B}ackward-{O}racle-{M}atching algorithm.
\newblock {\em International Journal of Foundations of Computer Science},
  20(06):967--984, 2009.

\bibitem{Faro2013}
S.~Faro and T.~Lecroq.
\newblock The {E}xact {O}nline {S}tring {M}atching {P}roblem: {A} {R}eview of
  the {M}ost {R}ecent {R}esults.
\newblock {\em ACM Comput. Surv.}, 45(2):13:1--13:42, Mar. 2013.

\bibitem{Franek2005}
F.~Franek, C.~G. Jennings, and W.~F. Smyth.
\newblock A simple fast hybrid pattern-matching algorithm.
\newblock In {\em Combinatorial Pattern Matching}, pages 288--297. Springer,
  2005.

\bibitem{Guibas1981}
L.~Guibas and A.~Odlyzko.
\newblock String overlaps, pattern matching, and nontransitive games.
\newblock {\em Journal of Combinatorial Theory, Series A}, 30(2):183 -- 208,
  1981.

\bibitem{Horspool1980}
R.~N. Horspool.
\newblock Practical fast searching in strings.
\newblock {\em Software: Practice and Experience}, 10(6):501--506, 1980.

\bibitem{Karp1987}
R.~M. Karp and M.~O. Rabin.
\newblock Efficient randomized pattern-matching algorithms.
\newblock {\em IBM Journal of Research and Development}, 31(2):249--260, 1987.

\bibitem{Knuth1977}
D.~E. Knuth, J.~H. Morris, Jr, and V.~R. Pratt.
\newblock Fast pattern matching in strings.
\newblock {\em SIAM journal on computing}, 6(2):323--350, 1977.

\bibitem{Mahmoud1997}
H.~M. Mahmoud, R.~T. Smythe, and M.~R\'egnier.
\newblock Analysis of {B}oyer-{M}oore-{H}orspool string-matching heuristic.
\newblock {\em Random Struct. Algorithms}, 10(1-2):169--186, 1997.

\bibitem{Marschall2012}
T.~Marschall, I.~Herms, H.~Kaltenbach, and S.~Rahmann.
\newblock {P}robabilistic {A}rithmetic {A}utomata and {T}heir {A}pplications.
\newblock {\em IEEE/ACM Trans. Comput. Biol. Bioinformatics}, 9(6):1737--1750,
  Nov. 2012.

\bibitem{Marschall2008}
T.~Marschall and S.~Rahmann.
\newblock {P}robabilistic {A}rithmetic {A}utomata and {T}heir {A}pplication to
  {P}attern {M}atching {S}tatistics.
\newblock In P.~Ferragina and G.~M. Landau, editors, {\em {C}ombinatorial
  {P}attern {M}atching}, volume 5029 of {\em Lecture Notes in Computer
  Science}, pages 95--106. Springer Berlin Heidelberg, 2008.

\bibitem{Marschall2010}
T.~Marschall and S.~Rahmann.
\newblock {E}xact {A}nalysis of {H}orspool’s and {S}unday’s {P}attern
  {M}atching {A}lgorithms with {P}robabilistic {A}rithmetic {A}utomata.
\newblock In A.-H. Dediu, H.~Fernau, and C.~Mart\'{\i}n-Vide, editors, {\em
  Language and Automata Theory and Applications}, volume 6031 of {\em Lecture
  Notes in Computer Science}, pages 439--450. Springer Berlin Heidelberg, 2010.

\bibitem{Marschall2011}
T.~Marschall and S.~Rahmann.
\newblock {A}n {A}lgorithm to {C}ompute the {C}haracter {A}ccess {C}ount
  {D}istribution for {P}attern {M}atching {A}lgorithms.
\newblock {\em Algorithms}, 4(4):285, 2011.

\bibitem{Regnier1998}
M.~R\'egnier and W.~Szpankowski.
\newblock {C}omplexity of {S}equential {P}attern {M}atching {A}lgorithms.
\newblock In M.~Luby, J.~D. Rolim, and M.~Serna, editors, {\em Randomization
  and Approximation Techniques in Computer Science}, volume 1518 of {\em
  Lecture Notes in Computer Science}, pages 187--199. Springer Berlin
  Heidelberg, 1998.

\bibitem{Smythe2001}
R.~T. Smythe.
\newblock {T}he {B}oyer-{M}oore-{H}orspool heuristic with {M}arkovian input.
\newblock {\em Random Struct. Algorithms}, 18(2):153--163, 2001.

\bibitem{Sunday1990}
D.~M. Sunday.
\newblock A very fast substring search algorithm.
\newblock {\em Communications of the ACM}, 33(8):132--142, 1990.

\bibitem{Thathoo2006}
R.~Thathoo, A.~Virmani, S.~Sai~Lakshmi, N.~Balakrishnan, and K.~Sekar.
\newblock {TVSBS}: A fast exact pattern matching algorithm for biological
  sequences.
\newblock {\em Current Science}, 91(1):47--53, 2006.

\bibitem{Tsai2006}
T.-H. Tsai.
\newblock {A}verage {C}ase {A}nalysis of the {B}oyer-{M}oore {A}lgorithm.
\newblock {\em Random Struct. Algorithms}, 28(4):481--498, July 2006.

\bibitem{Wu1994}
S.~Wu, U.~Manber.
\newblock A fast algorithm for multi-pattern searching.
\newblock {\em Tech. Report TR-94-17, CS Dept., University of Arizona}, 1994.

\bibitem{Yao1979}
A.~C.-C. Yao.
\newblock The complexity of pattern matching for a random string.
\newblock {\em SIAM Journal on Computing}, 8(3):368--387, 1979.

\end{thebibliography}
\end{document}